\documentclass[11pt,reqno]{amsart}
\pdfoutput=1
\usepackage{amsmath,bbm,physics}
\usepackage{hyperref}
\hypersetup{citecolor=blue}
\usepackage{cleveref}
\usepackage[style=numeric,maxbibnames=99]{biblatex}
\DefineBibliographyExtras{UKenglish}{}
\usepackage{capt-of}

\usepackage{import}
\usepackage{packages}
\usepackage{definitions}
\title[Quantum Zeno Effect]{Optimal convergence rate in the quantum Zeno effect for open quantum systems in infinite dimensions}
\author[M\"obus, Rouz\'e]{Tim M\"obus$^{1,2}$, Cambyse Rouz\'e$^{1,2}$}
\email{tim.moebus@tum.de, rouzecambyse@gmail.com}
\address{$^1$ Department of Mathematics, Technical University of Munich}
\address{$^2$ Munich Center for Quantum Science and Technology (MCQST),  M\"unchen, Germany}

\begin{document}
	\begin{abstract}
		In open quantum systems, the quantum Zeno effect consists in frequent applications of a given quantum operation, e.g.~a measurement, used to restrict the time evolution (due e.g.~to decoherence) to states that are invariant under the quantum operation. In an abstract setting, the Zeno sequence is an alternating concatenation of a contraction operator (quantum operation) and a $C_0$-contraction semigroup (time evolution) on a Banach space. In this paper, we prove the optimal convergence rate $\cO(\tfrac{1}{n})$ of the Zeno sequence by proving explicit error bounds. For that, we derive a new Chernoff-type $\sqrt{n}$-Lemma, which we believe to be of independent interest. Moreover, we generalize the convergence result for the Zeno effect in two directions: We weaken the assumptions on the generator, inducing the Zeno dynamics generated by an unbounded generator and we improve the convergence to the uniform topology. Finally, we provide a large class of examples arising from our assumptions.
	\end{abstract}
	\maketitle
	\tableofcontents
\newpage

	\section{Introduction}\label{sec:intro}
		In its original form, the quantum Zeno effect is defined for closed finite quantum systems. \citeauthor{Misra.1977}  predicted that "an unstable particle which is continuously observed to see whether it decays will never be found to decay!" \cite[Abst.]{Misra.1977}. In a more general setup, frequent measurements  enable a change in the time evolution and convergence to the so called \textit{Zeno dynamics}. Experimentally, the Zeno effect is verified for instance in \cite{Itano.1990, Fischer.2001}. In addition to its theoretical value, the quantum Zeno effect is used in error correction schemes to suppress decoherence in open quantum systems \cites{Hosten.2006}{Franson.2004}{Beige.2000}{Barankai.2018}{Luchnikov.2017}. The idea is to frequently measure the quantum state and thereby force the evolution to remain within the code space. With an appropriate measurement, one can even decouple the system from its environment \cites{Facchi.2004}{Burgarth.2020} and show that appropriately encoded states can be protected from decoherence with arbitrary accuracy \cites{Dominy.2013}{Erez.2004}. Moreover, the quantum Zeno effect has been used in commercial atomic magnetometers \cite{Kominis.2009}.
		
		Introduced by \citeauthor{Beskow.1967} in 1967 and later named by \citeauthor{Misra.1977} after the greek philosopher Zeno of Elea, the quantum Zeno effect in its simplest form can be stated as follows: given a projective measurement $P$ and a unitary time-evolution generated by a Hamiltonian $H$ acting on a finite dimensional Hilbert space $\cH$ \cite{Misra.1977}: For $n\to\infty$
		\begin{equation}\label{eq:zeno-misra-sudarshan}
    		(Pe^{\frac{it}{n}H})^n\longrightarrow e^{it\,PHP}
		\end{equation}
	    Since the seminal works \cites{Beskow.1967}{Misra.1977}, the result was extended in many different directions (overviews can be found in \cites{Facchi.2008}{Schmidt.2004}{Itano.1990}). Recently, the convergence in \Cref{eq:zeno-misra-sudarshan} was proven in the strong topology for unbounded Hamiltonian under the weak assumption that $PHP$ is the generator of a $C_0$-semigroup \cite{Exner.2021}. Earlier approaches used the so called \textit{asymptotic Zeno condition} \cites{Schmidt.2004}{Exner.1989}{Misra.1977}, which assumes $(\1-P)e^{itH}P$ and $Pe^{itH}(\1-P)$ to be Lipschitz continuous at $t=0$. This condition is natural in the sense that it is related to the boundedness of the first moment of the Hamiltonian in the initial state and is efficiently verifiable in practice. With the works \cites{Burgarth.2020}{Mobus.2019}{Barankai.2018}, the quantum Zeno effect was generalized to open and infinite dimensional quantum systems equipped with general quantum operations and uniformly continuous time evolutions. Note that in open quantum systems, we are dealing with operators acting on the Banach space $\mathcal{T}_1(\cH)$ of trace-class operators. More recently, \citeauthor{Becker.2021} generalized
		the Zeno effect further and interpreted the Zeno sequence as a product formula consisting of a contraction $M$ (quantum operation) and a $C_0$-contraction semigroup (quantum time evolution) on an abstract Banach space. Under a condition of \textit{uniform power convergence} of the power sequence $\{M^k\}_{k\in\mathbb{N}}$ towards a projection $P$ and boundedness of $M\cL $ and $\cL M$, they proved a quantitative bound on the convergence rate \cite{Becker.2021}:
		\begin{align}\label{eq:Becker}
			\|(Me^{\frac{t}{n}\cL})^nx-e^{tP\cL P}Px\|=\cO(n^{-\frac{1}{3}}(\|x\|+\|\cL x\|))\,,
		\end{align} 
		for $n\rightarrow\infty$ and all $x\in\cD(\cL)$. However, the optimality of \eqref{eq:Becker} was left open.\footnote{Note that we found an inconsistency in the proof of \cite[Lem.~2.1]{Zagrebnov.2017} (see \cite{Zagrebnov.2022}), which slightly reduces the convergence rate found in \Cref{eq:Becker} (more details are given in \Cref{sec:alternative-chernoff-lemma-trotter-product-formula}).}
		
		\subsubsection*{Main contributions:}
			In this paper, we achieve the optimal convergence rate $\mathcal{O}(n^{-1})$ of the Zeno sequence consistent with the finite-dimensional case \cite{Burgarth.2020} by providing an explicit bound which recently attracted interest in finite closed quantum systems \cite[Thm.~1]{Burgarth.2021}. Moreover, we generalize the results of \cite{Becker.2021} in two complementary directions:
			
			In \Cref{thm:spectral-gap} below, we assume a special case of the \textit{uniform power convergence} assumption on $M$, that is $\|M^n-P\|\leq\delta^n$ for some $\delta\in(0,1)$, and weaken the assumption on the semigroup to the \textit{uniform asymptotic Zeno condition} inherited from the unitary setting of \cite{Schmidt.2004}: for $t\rightarrow0$
			\begin{equation*}
				\norm{(\1-P) e^{t\cL}P}_\infty=\cO(t)\quad\text{and}\quad \norm{Pe^{t\cL}(\1-P)}_\infty=\cO(t).
			\end{equation*}
			Therefore, we prove the convergence of a non-trivial Zeno sequence in open quantum systems to a Zeno dynamics described by a possibly unbounded generator.\\
			Second, \Cref{thm:spectral-gap-uniform} is stated under slightly weaker assumptions as Theorem 3 in \cite{Becker.2021} and improves the result to the optimal convergence rate and to the uniform topology.
		
			In order to achieve these results, we prove a modified Chernoff $\sqrt{n}$-Lemma in \Cref{lem:improved-chernoff}, find a quantitative convergence rate for $\operatorname{exp}\big({nP(e^{\frac{1}{n}t\cL}-\1)P}\big)P-\operatorname{exp}(tP\cL P)P$ as $n \rightarrow\infty$, where $P\cL P$ is possibly unbounded, and prove the upper semicontinuity of parts of the spectrum of $Me^{t\cL}$ under tight assumptions.
		
		\subsubsection*{Organization of the paper:}
			In \Cref{sec:prelim}, we provide a short recap on bounded and unbounded operator theory. We expose our main results in \Cref{sec:assumptions-results}. \Cref{sec:alternative-chernoff-lemma-trotter-product-formula} deals with the modified Chernoff $\sqrt{n}$-Lemma and some of its implications as regards to Trotter-Kato's product formula. Then, we prove our main theorems under the weakest possible assumptions on the $C_0$-semigroup in \Cref{sec:unbounded-generator-zeno-subspace}, and under the weakest possible assumptions on $M$ in \Cref{sec:finitely-many-eigenvalues}. Our results are illustrated by three large classes of examples in finite and infinite dimensional quantum systems in \Cref{sec:applications}. Finally, we discuss some remaining open questions in \Cref{sec:discussion}.

	\section{Preliminaries}\label{sec:prelim}
		Let $(\cX,\|\cdot\|)$ be a Banach space and $(\cB(\cX),\|\cdot\|_{\infty})$ be the associated space of bounded linear operators over $\C$ equipped with the operator norm,  i.e.~$\|C\|_{\infty}\coloneqq\sup_{x\in\cX\backslash\{0\}}\frac{\|Cx\|}{\|x\|}$, and the identity $\1\in\cB(\cX)$. By a slight abuse of notation, we extend all densely defined and bounded operators by the \textit{bounded linear extension theorem} to bounded operators on $\cX$ \cite[Thm.~2.7-11]{Kreyszig.1989}.  
		A sequence $(C_k)_{k\in\N}\subset\cB(\cX)$ converges uniformly to $C\in\cB(\cX)$ if $\lim_{k\rightarrow\infty}\|C_k-C\|_\infty=0$ and strongly if $\lim_{k\rightarrow\infty}\|C_kx-Cx\|=0$ for all $x\in\cX$. The integral over bounded vector-valued maps, e.g.~$[a,b]\rightarrow\cX$ or $[a,b]\rightarrow \cB(\cX)$ with $a<b$, is defined by the \textit{Bochner integral}, which satisfies the triangle inequality, is invariant under linear transformations, and satisfies the \textit{fundamental theorem of calculus} if the map is continuously differentiable \cites[Sec.~3.7-8]{Hille.2000}[Sec.~A.1-2]{Liu.2015}.\\
		We define the \textit{resolvent set} of $C\in\cB(\cX)$ by $\rho(C)\coloneqq\{z\in\C\;|\;(z-C)\text{ bijective}\}$ and its \textit{spectrum} by $\sigma(C)\coloneqq\C\backslash\rho(C)$. Then, we define the \textit{resolvent} of $C$ by $R(z,C)\coloneqq(z-C)^{-1}$ for all $z\in\rho(C)$. Note that the resolvent is continuous in $z$ \cite[Thm.~3.11]{Kato.1995} and thereby uniformly bounded on compact intervals. A point $\lambda\in\sigma(C)$ is called isolated if there is a neighbourhood $V\subset\C$ of $\lambda$ such that $V\cap\sigma(C)=\lambda$. The spectrum is separated by a curve $\Gamma:I\rightarrow\rho(C)$ if $\sigma(C)=\sigma_1\cup\sigma_2$, the subsects $\sigma_1$, $\sigma_2$ have disjoint neighbourhoods $V_1$, $V_2$, and the curve is simple, closed, rectifiable\footnote{A simple, closed, and rectifiable curve is continuous, has finite length, joins up, and does not cross itself.}, and encloses one of the neighbourhoods without intersecting $V_1$ and $V_2$ \cite[Sec.~III§6.4]{Kato.1995}. This definition can be generalized to a finite sum of subsets separated by curves. An example for a simple, closed, and rectifiable curve is the parametrization of the boundary of a complex ball around the origin with radius $r$. We denote the open ball by $\D_r$ and its boundary by $\partial\D_r$.\\
		We define the \textit{spectral projection} $P\in\cB(\cX)$ with respect to a separated subset of $\sigma(C)$ enclosed by a curve $\Gamma$ via the \textit{holomorphic functional calculus} \cite[Thm.~2.3.1-3]{Simon.2015}
		\begin{equation}\label{eq:spectral-projection}
			P=\frac{1}{2\pi i}\oint_{\Gamma} R(z,C)dz.
		\end{equation}
		As the name suggests, $P\in\cB(\cX)$ satisfies the projection property, $P^2=P$. We denote the complementary projection $\1-P$ by $P^\perp$. If $\Gamma$ encloses the isolated one-point subset $\{\lambda\}\subset\sigma(C)$, then $CP=\lambda P+N$ with the \textit{quasinilpotent operator} \cite[Thm.~2.3.5]{Simon.2015}
		\begin{equation}\label{eq:quasinilpotent-operator}
			N\coloneqq\frac{1}{2\pi i}\oint_{\Gamma}(z-\lambda)R(z,C)dz\in\cB(\cX)\,.
		\end{equation}
		A quasinilpotent operator $N\in\cB(\cX)$ is characterized by $\lim\limits_{n\rightarrow\infty}\|N^n\|_\infty^{\frac{1}{n}}=0$.
		
		The time evolution in the Zeno sequence is described by a $C_0$-contraction semigroup \cite[Def.~5.1]{Engel.2000}. A family of operators $(T_t)_{t\geq0}\subset\cB(\cX)$ is called a $C_0$-semigroup if the family satisfies the semigroup properties, namely: (i) $T_tT_s=T_{t+s}$ for all $t,s\geq0$, (ii) $T_0=\1$, and (iii) $t\mapsto T_t$ is strongly continuous in $0$. The generator of a $C_0$-semigroup is a possibly unbounded operator defined by taking the strong derivative in $0$, that is
		\begin{equation*}
			\cL x\coloneqq\lim\limits_{h\downarrow 0}\frac{T(h)x-x}{h}
		\end{equation*}
		for all $x\in\cD(\cL)\coloneqq\{x\in\cX\;|\;t\mapsto T_tx\text{ differentiable}\}$ and is denoted by $(\cL,\cD(\cL))$ (cf.~\cite[Ch.~II.1]{Engel.2000}). Since the generator defines the $C_0$-semigroup uniquely, we denote it by $(e^{t\cL})_{t\geq0}$, although the series representation of the exponential is not well-defined (cf.~\cite[Sec.~IX§1]{Kato.1995}). A semigroup is called contractive if $\sup_{t\ge 0}\|e^{t\cL}\|_\infty\leq 1$. Note that the linear combination and concatenation of two unbounded operators $(\cK,\cD(\cK))$ and $(\cL,\cD(\cL))$ is defined on $\cD(\cK+\cL)=\cD(\cK)\cap\cD(\cL)$ and $\cD(\cK\cL)=\cL^{-1}(\cD(\cK))$ (cf.~\cite[Sec.~III§5.1]{Kato.1995}). The following lemma summarizes some properties of $C_0$-semigroups which are used in this work:
		\begin{lem}[\texorpdfstring{\cite[Lem.~II.1.3]{Engel.2000}}{[8, Lem.~II.1.3]}]\label{lem:properties-semigroups}
			Let $(\cL,\cD(\cL))$ be the generator of the $C_0$-semigroup $(e^{t\cL})_{t\geq0}$ defined on $\cX$. Then the following properties hold:
			\begin{enumerate}
				\item If $x\in\cD(\cL)$, then $e^{t\cL}x\in\cD(\cL)$ and 
				\begin{equation*}
					\frac{\partial}{\partial t}e^{t\cL}x=e^{t\cL}\cL x=\cL e^{t\cL}x\qquad \text{for all $t\geq0$\,.}
				\end{equation*}
				\item For every $t\geq0$ and $x\in\cX$, one has
				\begin{equation*}
					\int_{0}^{t}e^{\tau\cL}xd\tau\in\cD(\cL)\,.
				\end{equation*}
				\item For every $t\geq0$, one has
				\begin{align*}
					e^{t\cL}x-x&=\cL\int_{0}^{t}e^{\tau\cL}xd\tau\quad\text{ for all }x\in\cX\\
					&=\int_{0}^{t}e^{\tau\cL}\cL xd\tau\quad\text{ if }x\in\cD(\cL).
				\end{align*}
			\end{enumerate}
		\end{lem}
		Additionally, a general \textit{integral formulation} for the difference of two semigroups is discussed in the following \namecref{lem:integral-equation-semigroups} (cf.~\cite[Cor.~III.1.7]{Engel.2000}).
		\begin{lem}[Integral equation for semigroups]\label{lem:integral-equation-semigroups}
			Let $(\cL,\cD(\cL))$ and $(\cK,\cD(\cK))$ be the generators of two $C_0$-semigroups on a Banach space $\cX$ and $t\geq0$. Assume $\cD(\cL)\subset\cD(\cK)$ and $[0,t]\ni s\mapsto e^{s\cK}(\cK-\cL)e^{(t-s)\cL}x$ is continuous for all $x\in\cD(\cL)$. Then, for all $x\in\cD(\cL)$
			\begin{equation*}
				e^{t\cK}x-e^{t\cL}x=\int_{0}^{t}e^{s\cK}\left(\cK-\cL\right)e^{(t-s)\cL}xds.
			\end{equation*}
		\end{lem}
		\begin{proof}
			Assume $x\in\cD(\cL)$. In the first part, we follow the proof of Corollary III.1.7 in \cite{Engel.2000}. By definition, the vector-valued maps $[0,t]\ni s\mapsto e^{s\cK}x$ and $[0,t]\ni s\mapsto e^{(t-s)\cL}x$ are continuously differentiable and $\cD(\cL) $ is invariant under the second map. Then, Lemma B.16 in \cite{Engel.2000} shows
			\begin{equation*}
				\frac{\partial}{\partial s}e^{s\cK}e^{(t-s)\cL}x=e^{s\cK}(\cK-\cL)e^{(t-s)\cL}x.
			\end{equation*}
			By assumption the above derivative is continuous in $s\in[0,t]$, so that the fundamental theorem of calculus proves the integral equation.
		\end{proof}
		\begin{cor}\label{cor:integral-equation-semigroups-upper-bound}
			Let $(\cL,\cD(\cL))$ be the generator of a $C_0$-semigroup on a Banach space $\cX$, $t\geq0$, and $A\in\cB(\cX)$. Then, the unbounded operator $\cK=\cL+A$ defined on $(\cL+A,\cD(\cL))$ is the generator of a  $C_0$-semigroup. If additionally the semigroup $t\mapsto e^{t\cL}$ is quasi-contractive, i.e. $\|e^{t\cL}\|_\infty\leq e^{tw}$ for a $w\in\R$, then
		    \begin{equation*}
				\big\|e^{t(\cL+A)}-e^{t\cL}\big\|_\infty\leq e^{tw}(e^{t\norm{A}_\infty}-1).
			\end{equation*}
		\end{cor}
		\begin{proof}
		   By Theorem 13.2.1 and the Corollary afterwards in \cite{Hille.2000}, the operator $(\cL+A,\cD(\cL))$ is the generator of a quasi-contractive $C_0$-semigroup, i.e.~$\|e^{t(\cL+A)}\|_\infty\leq e^{t\tilde{w}}$ with $\tilde{w}=w+\|A\|_\infty$. Moreover, $[0,t]\ni s\mapsto e^{s\cK}Ae^{(t-s)\cL}x$ is continuous by \cite[Lem.~B.15]{Engel.2000} so that \Cref{lem:integral-equation-semigroups} shows\\
		    \begin{equation*}
		        \norm{e^{t(\cL+A)}-e^{t\cL}}_\infty\leq\int_{0}^{t}\norm{e^{s(\cL+A)}Ae^{(t-s)\cL}}_\infty ds\leq e^{tw}\norm{A}_\infty\int_{0}^{t} e^{s\norm{A}_\infty}ds=e^{tw}(e^{t\norm{A}_\infty}-1).
		    \end{equation*}
		\end{proof}
		
	\section{Main results}\label{sec:assumptions-results}
		In this section, we list the main results achieved in the paper. Since the quantum Zeno effect consists of a contraction operator $M\in\cB(\cX)$ (e.g.~measurement) and a $C_0$-semigroup (e.g.~quantum time evolution), the assumptions on the contraction operator and on the $C_0$-semigroup influence each other. We first start with weak assumptions on the semigroup and stronger assumption on the contraction operator:
		\begin{manualthm}{I}[stated as~\Caref{thm:spectral-gap} in main text]\label{thm:spectral-gap-prev}
		    Let $(\cL,\cD(\cL))$ be the generator of a $C_0$-contraction semigroup on $\cX$, $M\in\cB(\cX)$ a contraction, and $P$ a projection satisfying 
			\begin{equation}\label{unifpower1-prev}
				\norm{M^n-P}_\infty\leq\delta^n
			\end{equation}
	        for $\delta\in(0,1)$ and all $n\in\N$. Moreover, assume there is $b\geq0$ so that for all $t\ge 0$
			\begin{equation}\label{eq:thm1-asympzeno-prev}
				\norm{Pe^{t\cL}P^\perp}_\infty\leq tb\quad\text{and}\quad\norm{P^\perp e^{t\cL}P}_\infty\leq tb.
			\end{equation}
			If $(P\cL P,\cD(\cL P))$ is the generator of a $C_0$-semigroup, then for any $t\ge 0$ and all $x\in\cD((\cL P)^2)$ 
			\begin{equation*}
				\norm{\left(Me^{\frac{t}{n}\cL}\right)^nx-e^{tP\cL P}Px}\leq\frac{c(t,b)}{n}\left(\norm{x}+\norm{\cL P x}+\norm{(\cL P)^2x}\right)+\delta^n\norm{x}
			\end{equation*}
			for some constant $c(t,b)>0$ depending on $t$ and $b$, but independent of $n$ .
		\end{manualthm}
		\begin{rmk*}
			Note that the assumption on $M$ is a special case of the so-called \textit{uniform power convergence} assumption (q.v.~\Caref{unifpower-prev}) and the assumption (\ref{eq:thm1-asympzeno-prev}) on the $C_0$-semigroup is a generalization of the uniform \textit{asymptotic Zeno condition} which implies the convergence in the case of a unitary evolution frequently measured by a projective measurement \cite[Sec.~3.1]{Schmidt.2004}. Note that in that specific case, \cite{Exner.2021} recently managed to remove the asymptotic Zeno condition. Moreover, the assumption that $(P\cL P,\cD(\cL P))$ is a generator can be relaxed to the assumption that $P\cL P$ is closeable and its closure defines a generator (q.v.~remark after \Caref{lem:proofthm1-term3}). The famous \textit{Generation Theorem} by Hille and Yosida provides a sufficient condition under which $\overline{P\cL P}$ is a generator \cite[Thm.~3.5-3.8]{Engel.2000}.
		\end{rmk*}
	    The following example confirms the optimality of the achieved convergence rate.
		\begin{ex}
			Let $\left\{\ket{1},\ket{2},\ket{3}\right\}$ be an orthonormal basis of $\R^3$ and $\delta\in(0,1)$. We define, 
			\begin{equation*}
				\cL\coloneqq \ket{1}\bra{2}\quad\text{and}\quad M\coloneqq\ket{1}\bra{1}+\delta\ket{3}\bra{3}.
			\end{equation*}
			Then, $P=\ket{1}\bra{1}$, $(\1-P) M=\delta\ket{3}\bra{3}$, $M\cL=\cL M$, and $\cL^2=0=\cL P$, $\|M^n-P\|_\infty\leq\delta^n$. Using these properties, $Me^{\frac{t}{n}\cL}=M +\frac{t}{n}\cL$ and for $t\in[0,\infty)$
			\begin{align*}
			    \left(Me^{\frac{t}{n}\cL}\right)^n&=\left(M+\frac{t}{n}\cL\right)^n\\
			    &=\delta^n\ket{3}\bra{3}+\ket{1}\bra{1}+\frac{t}{n}\ket{1}\bra{2}\\
			    &=\left((\1-P)M\right)^ne^{t\cL}+Pe^{tP\cL P}+\frac{t}{n}P\cL.
			\end{align*}
			Therefore, 
			\begin{equation*}
				\norm{\left(Me^\frac{t}{n}\cL\right)^n-Pe^{tP\cL P}}_\infty=\max\{\tfrac{t}{n},\delta^n\},
			\end{equation*}
			which shows the optimality of our convergence rate in \Cref{thm:spectral-gap}.
		\end{ex}
        Beyond the proven asymptotics, we find explicit error bounds in \Cref{lem:proofthm1-term1}, \ref{lem:proofthm1-term2}, and \ref{lem:proofthm1-term3}, which simplify if $\cL$ is bounded to the following explicit convergence bound depending on the generator $\cL$, the projection $P$, the spectral gap $\delta$, and the time $t$:
		\begin{prop}\label{cor:explicit-bound-thm1}
			Let $\cL\in\cB(\cX)$ be the generator of a contractive uniformly continuous semigroup and $M\in\cB(\cX)$ a contraction satisfying
			\begin{equation*}
				\norm{M^n-P}_\infty\leq\delta^n
			\end{equation*}
			for a projection $P\in\cB(\cX)$, $\delta\in(0,1)$, and all $n\in\N$. Then, for all $t\geq0$ and $n\in\N$,
			\begin{align*}
		        \norm{\left(Me^{\frac{t}{n}\cL}\right)^n-e^{tP\cL P}P}_\infty\leq c_{p}\frac{t\|\cL\|_\infty}{n} &+\left(\frac{c_{p}+(1+e^{\tilde{b}})(1+c_{p}^2)}{2}\right)\frac{t^2\|\cL\|_\infty^2}{n}\\
		        &+\delta^n+\frac{2\delta}{1-\delta}\frac{e^{3t\|\cL\|_\infty c_{p}}}{n}
	    	\end{align*}
			where $c_{p}\coloneqq\|\1-P\|_\infty$ and $e^{\tilde{b}}=\sup_{s\in[0,t]}\|e^{sP\cL P}\|_\infty$.
		\end{prop}
		Note that the above proposition can be easily extended to the case of an unbounded generator with the assumption that $\cL M$ and $M\cL$ are densely defined and bounded. Another advantage of our setup is the freedom it provides for choosing the Banach space $\cX$, which allows us to treat open quantum systems ($\cX=\cT(\cH)$ the trace class operators over a Hilbert space) and closed quantum systems ($\cX=\cH$ a Hilbert space) on the same footing. In the case of finite dimensional closed quantum systems, \Cref{cor:explicit-bound-thm1} reduces to the following bound, which was independently proven in \cite[Thm.~1]{Burgarth.2021} (up to a change of the numerical constant in the quadratic term from $\tfrac{5}{2}$ to $2$):
	
		\begin{cor}\label{cor:explicit-bound-thm1-closed-sys}
			Let $\cH$ be a Hilbert space, $H\in\cB(\cH)$ be a hermitian operator, and $P\in\cB(\cH)$ a hermitian projection. Then,
			\begin{equation*}
				\norm{\left(Pe^{-i\frac{t}{n}H}\right)^n-e^{-itPH P}P}_\infty\leq\frac{1}{n}\left(t\norm{H}_\infty+\frac{5}{2}t^2\norm{H}_\infty^2\right)
			\end{equation*}
		\end{cor}
		To achieve the bound above, one inserts $\delta=0$ and $\cL=iH$ in \Cref{cor:explicit-bound-thm1}. Note that $PHP$ is hermitian and $\|e^{siPHP}\|_\infty=1$ for all $s\geq0$. 
		
		Next, we consider convergence rates under a slight weakening of the condition on the map $M$:

		\begin{cor}\label{cor:spectral-gap}
			Let $(\cL,\cD(\cL))$ be the generator of a $C_0$-contraction semigroup on $\cX$ and $M\in\cB(\cX)$ a contraction satisfying 
			\begin{equation*}
				\norm{M^n-P}_\infty\leq\tilde{c}\,\delta^n
			\end{equation*}
			for some projection $P$, $\delta\in(0,1)$, $\tilde{c}\ge 0$ and all $n\in\N$. Moreover, assume there is $b\geq0$ so that 
			\begin{equation*}
				\norm{Pe^{t\cL}P^\perp}_\infty\leq tb\quad\text{and}\quad\norm{P^\perp e^{t\cL}P}_\infty\leq tb.
			\end{equation*}
			If $(P\cL P,\cD(\cL P))$ is the generator of a $C_0$-semigroup, then there exists a constant $c>0$ and $n_0\in\N$ so that $\tilde{c}\delta^{n_0}\eqqcolon\tilde{\delta}<1$ and for all $x\in\cD((\cL P)^2)$
			\begin{equation*}
				\norm{\left(M^{n_0}e^{\frac{t}{n}\cL}\right)^nx-e^{tP\cL P}Px}\leq\frac{c}{n}\left(\norm{x}+\norm{\cL P x}+\norm{(\cL P)^2x}\right)+\tilde{\delta}^n\norm{x}.
			\end{equation*}
		\end{cor}
		The above \namecref{cor:spectral-gap} follows by the choice $n_0$ such that $\tilde{c}\delta^{n_0}<1$ and applying \Cref{thm:spectral-gap-prev} to $\tilde{M}\coloneqq M^{n_0}$. A more physically motivated result treating the same generalization as the \namecref{cor:spectral-gap} above is provided in the next result:
		
		\begin{manualprop}{II}[stated as~\Caref{prop:spectral-gap-uniform-norm-power-convergence} in main text]\label{prop:spectral-gap-uniform-norm-power-convergence-prev}
            Let $(\cL,\cD(\cL))$ be the generator of a $C_0$-contraction semigroup on $\cX$ and $M\in\cB(\cX)$ a contraction such that
            \begin{equation}\label{unifpowerc-prev}
				\norm{M^n-P}_\infty\leq\tilde{c}\,\delta^n
			\end{equation}
			for some projection $P$, $\delta\in(0,1)$ and $\tilde{c}\ge 0$. Moreover, we assume that there is $b\geq0$ so that 
			\begin{equation}\label{key-prev}
				\norm{Pe^{t\cL}P^\perp}_\infty\leq tb,\quad\norm{M^\perp e^{t\cL}-M^\perp}_\infty\leq tb,\quad\text{and}\quad\norm{P^\perp e^{t\cL}P}_\infty\leq tb
			\end{equation}
			where $M^\perp=(\1-P)M$. If $(P\cL P,\cD(\cL P))$ generates a $C_0$-semigroup, then there is an $\epsilon>0$ such that for all $t\geq0$, $n\in\N$ satisfying $t\in[0,n\epsilon]$, $\tilde{\delta}\in(\delta,1)$, and $x\in\cD((\cL P)^2)$ 
			\begin{equation*}
				\norm{\left(Me^{\frac{t}{n}\cL}\right)^nx-e^{tP\cL P}Px}\leq\frac{c_1(t,b,\tilde{\delta}-\delta)}{n}\left(\norm{x}+\norm{\cL P x}+\norm{(\cL P)^2x}\right)+c_2(\tilde{c},\tilde{\delta}-\delta)\tilde{\delta}^n\norm{x}\,,
			\end{equation*}
		    for some constants $c_1,c_2\ge 0$ depending on $t$, $b$, the difference $\tilde{\delta}-\delta$, and $\tilde{c}$.
		\end{manualprop}
		
		As in \Cref{cor:explicit-bound-thm1}, we also get a more explicit bound in the case of bounded generators in the following proposition:
		\begin{prop}\label{cor:explicit-bound-prop}
			Let $\cL\in\cB(\cX)$ be the generator of a contractive uniformly continuous semigroup and $M\in\cB(\cX)$ a contraction satisfying
			\begin{equation*}
				\norm{M^n-P}_\infty\leq \tilde{c}\delta^n
			\end{equation*}
			for a projection $P\in\cB(\cX)$, $\delta\in(0,1)$, $\tilde{c}>1$, and all $n\in\N$. Then there is $\epsilon>0$ such that for all $t\geq0$, $n\in\N$ satisfying $t\in[0,n\epsilon]$, and $\tilde{\delta}\in(\delta,1)$
			\begin{align*}
	            \norm{\left(Me^{\frac{t}{n}\cL}\right)^n-e^{tP\cL P}P}_\infty\leq\frac{tc_{p}\|\cL\|_\infty}{n}&+\frac{c_{p}+(1+e^{\tilde{b}})(1+c_{p}^2)}{2}\frac{t^2\|\cL\|_\infty^2}{n}\\
	            &+\frac{2\tilde{c}}{\tilde{\delta}-\delta}\tilde{\delta}^n+\frac{2\tilde{\delta}}{1-\tilde{\delta}}\frac{e^{\frac{6tc_{p}\tilde{c}\|\cL\|_\infty}{\tilde{\delta}-\delta}}}{n}
	        \end{align*}
	        where $c_{p}\coloneqq\|\1-P\|_\infty$ and $e^{\tilde{b}}=\sup_{s\in[0,t]}\|e^{sP\cL P}\|_\infty$.
		\end{prop}
		Finally, we extend the assumption on $M$ (q.v.~\Caref{unifpower1-prev}, (\ref{unifpowerc-prev})) to the \textit{uniform power convergence} introduced in \cite{Becker.2021}. Let $\{P_j,\lambda_j\}_{j=1}^J$ be a set of projections satisfying $P_jP_k=1_{j=k} P_j$ and associated eigenvalues on the unit circle $\partial\D_1$. Then, $M$ is called uniformly power convergent with rate $\delta\in(0,1)$ if $M^n-\sum_{j=1}^{J}\lambda^n_jP_j=\cO(\delta^n)$ uniformly for $n\rightarrow\infty$. To prove our result in this case, we also need to assume that $M\cL$ and $\cL P_\Sigma$ with $P_\Sigma\coloneqq\sum_{j=1}^{J}P_j$ are densely defined and bounded (cf.~\cite{Becker.2021}) and $P_j$ is a contraction for all $j\in\{1,...,J,\Sigma\}$:
		
        \begin{manualthm}{III}[stated as~\Caref{thm:spectral-gap-uniform} in main text]\label{thm:spectral-gap-uniform-prev}
		    Let $(\cL,\cD(\cL))$ be the generator of a $C_0$-contraction semigroup on $\cX$ and $M\in\cB(\cX)$ a contraction satisfying the following uniform power convergence: there is $\tilde{c}>0$ so that
			\begin{equation}\label{unifpower-prev}
			    \norm{M^n-\sum_{j=1}^{J}\lambda^n_jP_j}_\infty\leq\tilde{c}\,\delta^n
		    \end{equation}
		    for a set of projections $\{P_j\}_{j=1}^J$ satisfying $P_jP_k=1_{j=k}P_j$, eigenvalues $\{\lambda_j\}_{j=1}^J\subset\partial\D_1$, and a rate $\delta\in(0,1)$. For $P_\Sigma\coloneqq\sum_{j=1}^{J}P_j$, we assume that $M\cL$ and $\cL P_\Sigma$ are densely defined and bounded by $b\geq0$ and $\|P_j\|_\infty=1$ for all $j\in\{1,...,J,\Sigma\}$. Then, there is an $\epsilon>0$ such that for all $n\in\N$, $t\geq0$ satisfying $t\in[0,n\epsilon]$, and $\tilde{\delta}\in(\delta,1)$
			\begin{equation*}
				\norm{\left(Me^{\frac{t}{n}\cL}\right)^n-\sum_{j=1}^{J}\lambda_j^ne^{tP_j\cL P_j}P_j}_\infty\leq\frac{c_1 }{n}+c_2\tilde{\delta}^n\,,
			\end{equation*}
			for some constants $c_1,c_2\ge 0$ depending on all involved parameters except from $n$.
        \end{manualthm}
        
		In comparison to Theorem 3 in \cite{Becker.2021}, \Cref{thm:spectral-gap-uniform} achieves the optimal convergence rate and is formulated in the uniform topology under slightly weaker assumptions on the generator.
		\begin{rmk*}
			A natural way to weaken the above assumption is to assume that the power converges is in the strong topology (cf.~\cite[Thm.~2]{Becker.2021}).
		\end{rmk*}

	\section{Chernoff \texorpdfstring{$\sqrt{n}$-}{}Lemma and Trotter-Kato's Product Formula}\label{sec:alternative-chernoff-lemma-trotter-product-formula}
		In previous works \cite{Mobus.2019,Becker.2021}, Chernoff's $\sqrt{n}$-Lemma \cite[Lem.~2]{Chernoff.1968}, which we restate here, is used as a proof technique to approximate the Zeno product by a semigroup (q.v.~\Caref{eq:proofthm1-term2}).
		\begin{lem}[Chernoff \texorpdfstring{$\sqrt{n}$-}{square root }Lemma]\label{lem:chernoff}
			Let $C\in\cB(\cX)$ be a contraction. Then, $(e^{t(C-\1)})_{t\geq0}$ is a uniformly continuous contraction semigroup and for all $x\in\cX$
			\begin{equation*}
				\norm{C^nx-e^{n(C-\1)}x}\leq\sqrt{n}\norm{(C-\1)x}.
			\end{equation*}
		\end{lem}
		\begin{rmk*}
			In Lemma 2.1 in \cite{Zagrebnov.2017}, the dependence on $n$ is improved to $n^\frac{1}{3}$. This is crucial in the proof of the convergence rate in \cite[Lem.~5.4-5.5]{Becker.2021}. Unfortunately, we found an inconsistency in the proof of \cite[Lem.~2.1]{Zagrebnov.2017}, i.e.~Inequality 2.3 is not justified. An update and more Chernoff bounds can be found in \cite{Zagrebnov.2022}. Following the proof by \citeauthor{Becker.2021}, one can achieve a convergence rate of order $\tfrac{1}{\sqrt{n}}$ in the bounded generator case \cites[Thm.~1]{Becker.2021} and of order $\tfrac{1}{\sqrt[4]{n}}$ in the unbounded generator case \cites[Thm.~3]{Becker.2021}.
		\end{rmk*}
		In the case of the quantum Zeno effect (see \Caref{lem:proofthm1-term2} and \ref{lem:proofthm2-term2}) for bounded generators, the contraction $C$ is a vector-valued map $t\mapsto C(t)$ on $\cX$ satisfying $\|C(\tfrac{1}{n})-\1\|_\infty=\cO(n^{-1})$. By Chernoff's $\sqrt{n}$-Lemma 
		\begin{equation*}
			\norm{C^n(\tfrac{1}{n})-e^{n(C(\tfrac{1}{n})-\1)}}_\infty\leq\frac{1}{\sqrt{n}}.
		\end{equation*}
		Here, we chose the bounded generator case for sake of simplicity. Nevertheless, the argument can be extended to unbounded generator as well (see \Caref{lem:proofthm1-term2}, \ref{lem:proofthm2-term2}, and \ref{lem:approx-improved-chernoff}). Next, we prove a modified bound, which allows us to achieve the optimal rate in the quantum Zeno effect.
		\begin{lem}[Modified Chernoff Lemma]\label{lem:improved-chernoff}
			Let $C\in\cB(\cX)$ be a contraction and $n\in\N$. Then, $(e^{t(C-\1)})_{t\geq0}$ is a contraction semigroup and for all $x\in\mathcal{X}$
			\begin{equation*}
				\norm{\left(C^n-e^{n(C-\1)}\right)x}\leq\frac{n}{2}\norm{(C-\1)^2x}.
			\end{equation*}
		\end{lem}
		\begin{rmk*}
			At first glance, this seems to be worse than the original Chernoff $\sqrt{n}$-Lemmas in \cites{Chernoff.1968}. However, if $C$ is a vector-valued map satisfying $\|C(\tfrac{1}{n})-\1\|_\infty=\cO(n^{-1})$, then the modified Chernoff lemma gives
			\begin{equation*}
				\norm{C(\tfrac{1}{n})^n-e^{n(C(\text{\tiny{$\tfrac{1}{n}$}})-\1)}}_\infty\leq\frac{n}{2}\norm{C(\tfrac{1}{n})-\1}_{\infty}^2=\cO(n^{-1})
			\end{equation*}
			which is the key idea to prove the optimal convergence rate of the quantum Zeno effect for bounded generators and contractions $M$ satisfying the uniform power convergence (q.v.~\Caref{lem:proofthm1-term2}).
		\end{rmk*}
		\begin{proof}[Proof of \Cref{lem:improved-chernoff}]
			Similar to Chernoff's proof \cite[Lem.~2]{Chernoff.1968}, $(e^{t(C-\1)})_{t\geq0}$ is a contraction semigroup. We define $C_t\coloneqq (1-t)\1+tC=\1+t(C-\1)$ for $t\in[0,1]$, which itself is a contraction as a convex combination of contractions, and we use the fundamental theorem of calculus so that
			\begin{align*}
				\norm{\left(C^n-e^{n(C-\1)}\right)x}&\leq\int_{0}^{1}\norm{\frac{\partial}{\partial t}(C_t^ne^{(1-t)n(C-\1)})x}dt\\
				&\leq n\int_{0}^{1}\norm{C_t^{n-1}e^{(1-t)n(C-\1)}}_\infty\norm{(\1-C)(\1-C_t)x}dt\\
				&\leq\frac{n}{2}\norm{(C-\1)^2x},
			\end{align*}
			which proves the \namecref{lem:improved-chernoff}.
		\end{proof}
		In \cite[p.~241]{Chernoff.1968}, Chernoff proves the convergence of Trotter's product formula by approximating the product using the Chernoff $\sqrt{n}$-Lemma.
		For bounded generators, Chernoff's proof gives a convergence rate of order $n^{-\frac{1}{2}}$. Following his proof and using our modified Chernoff Lemma, we achieve the well-known optimal convergence rate of order $n^{-1}$ \cites[Thm.~2.11]{Hall.2015}[p.~1-2]{Neidhardt.2018}:		
		\begin{prop}[\texorpdfstring{\cite[Thm.~1]{Chernoff.1968}}{[6, Thm.~1]}]\label{prop:trotter-kato-product-formula}
			Let $F:\R_{\geq0}\rightarrow\cB(\cX)$ be a continuously differentiable function (in the uniform topology) satisfying $\sup_{t\in\R_{\geq0}}\|F(t)\|_\infty\leq1$. Assume that $F(0)=\1$ and denote the derivative at $t=0$ by $\cL\in\cB(\cX)$. Then, for all $t\geq0$
			\begin{equation*}
				\norm{F\left(\tfrac{t}{n}\right)^n-e^{t\cL}}_\infty\leq \norm{n\left(F\left(\tfrac{t}{n}\right)-\1\right)-t\cL}_\infty + \frac{n}{2}\norm{(F\left(\tfrac{t}{n}\right)-\1)}_\infty^2.
			\end{equation*}
		\end{prop}
		\begin{proof}
			The case $t=0$ is clear. For $t>0$, applying \Cref{lem:improved-chernoff}, we get
			\begin{align*}
				\norm{F\left(\tfrac{t}{n}\right)^n-e^{t\cL}}_\infty&\leq\norm{F\left(\tfrac{t}{n}\right)^n-e^{n\left(F\left(\text{\tiny{$\tfrac{t}{n}$}}\right)-\1\right)}}_\infty+\norm{e^{n\left(F\left(\text{\tiny{$\tfrac{t}{n}$}}\right)-\1\right)}-e^{t\cL}}_\infty\\
				&\leq \frac{n}{2}\norm{F\left(\tfrac{t}{n}\right)-\1}_\infty^2+\norm{e^{n\left(F\left(\text{\tiny{$\tfrac{t}{n}$}}\right)-\1\right)}-e^{t\cL}}_\infty.
			\end{align*}
			For the second term above, we apply \Cref{lem:integral-equation-semigroups}:
			\begin{align*}
				\norm{e^{n\left(F\left(\text{\tiny{$\tfrac{t}{n}$}}\right)-\1\right)}-e^{t\cL}}_\infty&=\int_{0}^{1}\norm{e^{sn\left(F\left(\text{\tiny{$\tfrac{t}{n}$}}\right)-\1\right)}\left(n\left(F\left(\tfrac{t}{n}\right)-\1\right)-t\cL\right)e^{(1-s)t\cL}}_\infty ds\\
				&\leq\norm{n\left(F\left(\tfrac{t}{n}\right)-\1\right)-t\cL }_\infty,
			\end{align*}
			where we use $\|F(\tfrac{t}{n})\|_\infty\leq1$ and that $e^{sn(F(\text{\tiny{$\tfrac{t}{n}$}})-\1)}$ is a contraction semigroup.
		\end{proof}
		Applying the proposition to the case of Trotter's product formula, we achieve the well-known optimal convergence rate for bounded generators on Banach spaces \cites[Thm.~2.11]{Hall.2015}[p.~1-2]{Neidhardt.2018}:
		\begin{cor}[\texorpdfstring{\cites[Thm.~2.11]{Hall.2015}}{[20, Thm.~2.11]}]\label{cor:trotter-kate-product-formula-convergence}
			Let $\cL_1$ and $\cL_2$ be bounded generators of two uniformly continuous contraction semigroups. Then, for $n\rightarrow\infty$
			\begin{equation*}
				\norm{\left(e^{\frac{1}{n}\cL_1}e^{\frac{1}{n}\cL_2}\right)^n-e^{\cL_1+\cL_2}}_\infty=\cO\left(\frac{1}{n}\right).
			\end{equation*}
		\end{cor}
		\begin{proof}
			We define $F(\tfrac{1}{n})\coloneqq e^{\frac{1}{n}\cL_1}e^{\frac{1}{n}\cL_2}$ for which 
			\begin{align*}
				F(\tfrac{1}{n})-\1&=e^{\frac{1}{n}\cL_1}\left(e^{\frac{1}{n}\cL_2}-\1\right)+e^{\frac{1}{n}\cL_1}-\1\\
				&=\left(e^{\frac{1}{n}\cL_1}-\1\right)\left(e^{\frac{1}{n}\cL_2}-\1\right)+e^{\frac{1}{n}\cL_2}-\1+e^{\frac{1}{n}\cL_1}-\1
			\end{align*}
			holds. Moreover,
			\begin{equation*}
				\norm{e^{\frac{1}{n}\cL_1}-\1}_\infty=\frac{1}{n}\norm{\int_{0}^{1}e^{\frac{\tau_1}{n}\cL_1}\cL_1d\tau_1}_\infty\leq\frac{1}{n}\norm{\cL_1}_\infty
			\end{equation*}
			and
			\begin{equation*}
				\norm{n\left(e^{\frac{1}{n}\cL_1}-\1\right)-\cL_1}_\infty=\frac{1}{n}\norm{\int_{0}^{1}\int_{0}^{1}\tau_1e^{\frac{\tau_1\tau_2}{n}\cL_1}\cL_1^2d\tau_2d\tau_1}_\infty\leq\frac{1}{2n}\norm{\cL_1}^2_\infty.
			\end{equation*}
			\begin{align*}
				\norm{n\left(F(\tfrac{1}{n})-\1\right)-\cL_1-\cL_2}_\infty\leq\frac{1}{n}\left(\norm{\cL_1}_\infty\norm{\cL_2}_\infty+2\norm{\cL_1}^2_\infty+2\norm{\cL_2}^2_\infty\right)
			\end{align*}
			and the statement follows from \Cref{prop:trotter-kato-product-formula}.
		\end{proof}

	\section{Strongly Continuous Zeno Dynamics}\label{sec:unbounded-generator-zeno-subspace}
		We proceed with the statement and proof of our first main result, namely \Cref{thm:spectral-gap}, which we restate here for sake of clarity of conciseness:
		
		\begin{thm}\label{thm:spectral-gap}
		    Let $(\cL,\cD(\cL))$ be the generator of a $C_0$-contraction semigroup on $\cX$, $M\in\cB(\cX)$ a contraction, and $P$ a projection satisfying 
			\begin{equation}\label{unifpower1}
				\norm{M^n-P}_\infty\leq\delta^n
			\end{equation}
	        for $\delta\in(0,1)$ and all $n\in\N$. Moreover, assume there is $b\geq0$ so that for all $t\ge 0$
			\begin{equation}\label{eq:thm1-asympzeno}
				\norm{Pe^{t\cL}P^\perp}_\infty\leq tb\quad\text{and}\quad\norm{P^\perp e^{t\cL}P}_\infty\leq tb.
			\end{equation}
			If $(P\cL P,\cD(\cL P))$ is the generator of a $C_0$-semigroup, then for any $t\ge 0$ and all $x\in\cD((\cL P)^2)$ 
			\begin{equation*}
				\norm{\left(Me^{\frac{t}{n}\cL}\right)^nx-e^{tP\cL P}Px}\leq\frac{c(t,b)}{n}\left(\norm{x}+\norm{\cL P x}+\norm{(\cL P)^2x}\right)+\delta^n\norm{x}
			\end{equation*}
		for a constant $c(t,b)>0$ depending on $t$ and $b$, but independent of $n$ .
		\end{thm}
		\subsection{Proof of \texorpdfstring{\Cref{thm:spectral-gap} }{Theorem 3.1}}\label{subsec:proofthm1}~\\
			We assume for sake of simplicity that $t=1$, and split our proof in three parts:
			\begin{align}
				\norm{\left(Me^{\frac{1}{n}\cL}\right)^nx-e^{P\cL P}Px}\leq\;&\norm{\left(Me^{\frac{1}{n}\cL}\right)^nx-\left(Pe^{\frac{1}{n}\cL}P\right)^nx}\label{eq:proofthm1-term1}\\	+\;&\norm{\left(Pe^{\frac{1}{n}\cL}P\right)^nx-e^{nP(e^{\frac{1}{n}\cL}-\1)P}Px}\label{eq:proofthm1-term2}\\
				+\;&\norm{e^{nP(e^{\frac{1}{n}\cL}-\1)P}Px-e^{P\cL P}Px}\label{eq:proofthm1-term3}
			\end{align}
			for all $x\in\cD((\cL P)^2)$.
			
			\subsubsection{Upper bound on \Cref{eq:proofthm1-term1}:}
				The following \namecref{lem:proofthm1-term1} uses similar proof strategies as Lemma 3 in \cites{Burgarth.2020} and extends the result to infinite dimensions in the strong topology.
				\begin{lem}\label{lem:proofthm1-term1} 
					Let $(\cL,\cD(\cL))$ be the generator of a $C_0$-contraction semigroup on $\cX$ and $M\in\cB(\cX)$ a contraction satisfying the assumptions in \Cref{thm:spectral-gap}. Then, for all $x\in\cX$
					\begin{equation*}
						\norm{\left(Me^{\frac{1}{n}\cL}\right)^nx-\left(Pe^{\frac{1}{n}\cL}P\right)^nx}\leq\left(\delta^n+\frac{b}{n}+\frac{1}{n}\frac{b(2+b)(\delta-\delta^{n})}{1-\delta}e^{2b}\right)\norm{x}.
					\end{equation*}
				\end{lem}
				The proof of the above \namecref{lem:proofthm1-term1} relies on a counting method: more precisely, we need to count the number of transitions in a binary sequence. This is related to the \textit{urn problem}, where $k$ indistinguishable balls are placed in $l$ distinguishable urns \cite[Chap.~1.9]{Stanley.1986}. Then, there are 
				\begin{equation}\label{eq:urn-problem}
					\binom{k-1}{l-1}
				\end{equation}
				possibilities to distribute the balls so that each urn contains at least one ball. 
				\begin{defi}\label{defi:counting-patterns}
					Let $S=\{A,B\}$, $j,n,k\in\N$, and $n\geq1$. We define
					\begin{align*}
						S_{n,k}&\coloneqq\{s\in S^n\;|\;\text{$A$ appears $k$ times in $s$}\}\\
						N(j,n,k)&\coloneqq\#\{s\in S_{n,k}\;|\;\text{$s$ includes $j$ transitions $AB$ or $BA$}\}.
					\end{align*}
				\end{defi}
				In words, $N(j,n,k)$ counts the number of sequences consisting of $k$ $A$'s and $n-k$ $B$'s with the restriction that $A$ changes to $B$ or vice versa $j$ times.
				\begin{ex}
					Let $S=\{A,B\}$, $n=4$, and $k=2$. Then,
					\begin{align*}
						N(0,n,k)&=\#\emptyset & &=0\\
						N(1,n,k)&=\#\{AABB,BBAA\} & &=2\\
						N(2,n,k)&=\#\{ABBA,BAAB\} & &=2\\
						N(3,n,k)&=\#\{ABAB,BABA\} & &=2.
					\end{align*}
				\end{ex}
				\begin{lem}\label{lem:counting-patterns-binary-sequences}
					Let $S=\{A,B\}$ and $n,k,j\in\N$ with $k\leq n$. Then,
					\begin{equation*}
						N(j,n,k)=\begin{cases}
							2\binom{n-k-1}{l-1}\binom{k-1}{l-1}& \text{if }j=2l-1\\
							\frac{n-2l}{l}\binom{n-k-1}{l-1}\binom{k-1}{l-1}& \text{if }j=2l\\
							1_{k\in\{0,n\}}& \text{if }j=0
						\end{cases}
					\end{equation*}
					for $j\in\{0,...,2\min\{k,n-k\}-1_{2k=n}\}$. Otherwise $N(j,n,k)=0$.
				\end{lem}
				\begin{proof}
					If $j\geq2\min\{n-k,k\}-1_{2k=n}$, then $N(j,n,k)=0$ by \Cref{defi:counting-patterns}. Next we assume that $j=0$, the only possible sequences are $A^n$ ($k=n$) and $B^n$ ($k=0$) so that $N(0,n,k)=1_{k\in\{0,n\}}$. In the following, we assume that $1\leq j\leq2\min\{n-k,k\}-1_{2k=n}$, then there is a $s\in S_{n,k}$ so that $s$ includes exactly $j$ transitions $AB$ or $BA$ so that $N(j,n,k)>0$. In the odd case $j=2l-1$ for $l\in\{1,...,\min\{k,n-k\}\}$, the element $s$ is constructed by $l$ blocks of $A$'s and $l$ blocks of $B$'s:
					\begin{align*}
						s&=\underbrace{A...A}_{1}\overbrace{B...B}^{1}\underbrace{A...A}_{2}B\quad.\quad.\quad.\quad A\overbrace{B...B}^{l-1}\underbrace{A...A}_{l}\overbrace{B...B}^{l},\\
						s&=\underbrace{B...B}_{1}\overbrace{A...A}^{1}\underbrace{B...B}_{2}A\quad.\quad.\quad.\quad B\overbrace{A...A}^{l-1}\underbrace{B...B}_{l}\overbrace{A...A}^{l}.
					\end{align*}
					Identifying these blocks with distinguishable urns and the elements $A$ and $B$ with indistinguishable balls (q.v.~\Caref{eq:urn-problem}), the task is to count the possibilities of placing $k$ $A$'s in $l$ urns and vice versa $n-k$ $B$'s in $l$ urns with the additional assumption that each urn must contain at least one $A$ or one $B$. By changing the roles of $A$ and $B$, we get twice the number of possible combinations. Therefore, one of the \textit{Twelvefold Ways} \cite[Chap.~1.9]{Stanley.1986} shows
					\begin{equation*}
						N(j,n,k)=2\binom{n-k-1}{l-1}\binom{k-1}{l-1}.
					\end{equation*}
					In the even case $j=2l$ for $l\in\{1,...,\min\{k,n-k\}-1_{2k=n}\}$, we argue similarly to the odd case. The only difference is that s is constructed by $l+1$ blocks of $A$'s and $l$ blocks of $B$'s or vice versa:
					\begin{align*}
						s&=\underbrace{A...A}_{1}\overbrace{B...B}^{1}\underbrace{A...A}_{2}B\quad.\quad.\quad.\quad A\overbrace{B...B}^{l-1}\underbrace{A...A}_{l}\overbrace{B...B}^{l}\underbrace{A...A}_{l+1},\\
						s&=\underbrace{B...B}_{1}\overbrace{A...A}^{1}\underbrace{B...B}_{2}B\quad.\quad.\quad.\quad A\overbrace{A...A}^{l-1}\underbrace{B...B}_{l}\overbrace{A...A}^{l}\underbrace{B...B}_{l+1}.
					\end{align*}
					Then the \textit{Twelvefold Ways} \cite[Chap.~1.9]{Stanley.1986} proves the statement by\\
					\begin{minipage}[b]{0.9\textwidth}
						\begin{align*}
							N(j,n,k)&=\binom{n-k-1}{l-1}\binom{k-1}{l}+\binom{n-k-1}{l}\binom{k-1}{l-1}\\
							&=\frac{n-2l}{l}\binom{n-k-1}{l-1}\binom{k-1}{l-1}.
						\end{align*}
					\end{minipage}
				\end{proof}
				With the help of this counting method, we are ready to prove \Cref{lem:proofthm1-term1}. In what follows, we identify the couple $(A,B)$ with the product $AB$ by slight abuse of notations.
				\begin{proof}[Proof of \Cref{lem:proofthm1-term1}]
					Assume w.l.o.g.~$P\neq 0$, then $MP=PM=P$ because for all $n\in\N$
					\begin{equation}\label{eq:eigenprojection}
						\|P-PM\|_\infty\leq\|(M^{n}-P)M\|_\infty+\|P-M^{n+1}\|_\infty\leq\delta^n+\delta^{n+1}
					\end{equation}
					and $\norm{P}_\infty\leq1$ holds by a similar argument because for all $n\in\N$
					\begin{equation}\label{eq:projection-contraction}
						\|P\|_\infty\leq\|M^n\|_\infty+\|P-M^n\|_\infty\leq 1+\delta^n.
					\end{equation}
					The main idea is to split $M=P+M^\perp$ with $M^\perp\coloneqq P^\perp M$ and order the terms after expanding the following polynomial appropriately. Let $A\coloneqq M^\perp e^{\frac{1}{n}\cL}$ and $B\coloneqq Pe^{\frac{1}{n}\cL}$ so that
					\begin{equation}\label{eq:proofthm1-polynomial-expansion}
						\left((P+M^\perp) e^{\frac{1}{n}\cL}\right)^n= B^n+\sum_{k=1}^{n-1}\sum_{s\in S_{n,k}}s+A^n
					\end{equation}
					where elements in $S_{n,k}$ are identified with sequences of concatenated operators and denoted by $s$. Then, we partition summands by the number of transitions from $A$ to $B$ or vice versa and use
					\begin{equation*}
						\norm{A}_\infty\le \delta\qquad \text{ and }\qquad 	\norm{AB}_\infty=\norm{M^\perp P^\perp e^{\frac{1}{n}\cL}PM e^{\frac{1}{n}\cL}}_\infty\leq\delta\frac{b}{n}.
					\end{equation*}
					The number of summands with $j$ transitions is equal to $N(j,n,k)$ given by \Cref{lem:counting-patterns-binary-sequences} for $j\in\{1,...,m\}$ and $m\coloneqq2\min\{k,n-k\}-1_{2k=n}$. Then, the inequality above shows 
					\begin{equation*}
						\begin{aligned}
							\norm{\left(Me^{\frac{1}{n}\cL}\right)^n-\left(Pe^{\frac{1}{n}\cL}\right)^n}_\infty&\leq\delta^n+\sum_{k=1}^{n-1}\sum_{j=1}^{m}\delta^{k}N(j,n,k)\left(\frac{b}{n}\right)^j\\
							&=\delta^n\begin{aligned}[t]
								&+\sum_{k=1}^{n-1}\sum_{l=1}^{\ceil*{\frac{m}{2}}}\delta^{k}2\binom{n-k-1}{l-1}\binom{k-1}{l-1}\left(\frac{b}{n}\right)^{2l-1}\\
								&+\sum_{k=1}^{n-1}\sum_{l=1}^{\floor*{\frac{m}{2}}}\delta^{k}\frac{b(n-2l)}{nl}\binom{n-k-1}{l-1}\binom{k-1}{l-1}\left(\frac{b}{n}\right)^{2l-1}
							\end{aligned}\\
							&\overset{(1)}{\leq}\delta^n+\frac{b(2+b)}{n}\sum_{l=1}^{\floor*{\frac{n}{2}}}\sum_{k=1}^{n-1}\delta^{k}\frac{n^{2l-2}}{(l-1)!^2}\left(\frac{b}{n}\right)^{2l-2}\\
							&=\delta^n+\frac{b(2+b)}{n}\frac{\delta-\delta^{n}}{1-\delta}\sum_{l=0}^{\floor*{\frac{n}{2}}-1}\frac{b^{2l}}{l!^2}\\
							&{\leq}\delta^n+\frac{1}{n}\frac{b(2+b)(\delta-\delta^{n})}{1-\delta}e^{2b}.
						\end{aligned}
					\end{equation*}
					In (1) above, we used the upper bound $\binom{n}{k}\leq\frac{n^k}{k!}$ to show
					\begin{equation*}
						\binom{n-k-1}{l-1}\binom{k-1}{l-1}\leq\frac{n^{2l-2}}{(l-1)!^2}.
					\end{equation*}
					Additionally, we increase the upper index to $\floor*{\frac{n}{2}}$, and upper bound
					\begin{equation*}
						2+\frac{b(n-2l)}{nl}\leq2+b.
					\end{equation*}
					Applying the assumptions again to
					\begin{equation*}
						\left(Pe^{\frac{1}{n}\cL}\right)^nx-\left(Pe^{\frac{1}{n}\cL}P\right)^nx=\left(Pe^{\frac{1}{n}\cL}P\right)^{n-1}Pe^{\frac{1}{n}\cL}P^\perp x
					\end{equation*}
					finishes the \namecref{lem:proofthm1-term1}:\\
					\begin{minipage}[b]{0.9\textwidth}
						\begin{equation*}
							\norm{\left(Me^{\frac{1}{n}\cL}\right)^nx-\left(Pe^{\frac{1}{n}\cL}P\right)^nx}\leq\left(\delta^n+\frac{b}{n}+\frac{1}{n}\frac{b(2+b)(\delta-\delta^{n})}{1-\delta}e^{2b}\right)\norm{x}.
						\end{equation*}
					\end{minipage}
				\end{proof}
				\begin{rmk*}
					By the counting method introduced above, we can approximate $(Me^{\frac{1}{n}\cL})^n$ by $(Pe^{\frac{1}{n}\cL}P)^n$, which is independent of $M^\perp$. In previous works \cites{Mobus.2019}{Becker.2021}, the operators considered in similar proof steps as \Cref{eq:proofthm1-term2} and (\ref{eq:proofthm1-term3}) depended on $M^\perp$.
				\end{rmk*}
			
			\subsubsection*{Upper bound on \Cref{eq:proofthm1-term2}:}
				In the next step, we apply our modified Chernoff \Cref{lem:improved-chernoff}:
				\begin{lem}\label{lem:proofthm1-term2}
					Let $(\cL,\cD(\cL))$ be the generator of a $C_0$-contraction semigroup on $\cX$ and $P\in\cB(\cX)$ be a projection. Assume that both operators satisfy the same assumption as in \Cref{thm:spectral-gap}. Then, for all $x\in\cD((\cL P)^2)$
					\begin{equation*}
						\norm{\left(Pe^{\frac{1}{n}\cL}P\right)^nx-e^{nP(e^{\frac{1}{n}\cL}-\1)P}Px}\leq\frac{1}{2n}\left(b^2\norm{x}+b\norm{\cL Px}+\norm{(\cL P)^2x}\right).
					\end{equation*}
				\end{lem}
				\begin{proof}
					The proof relies on the modified Chernoff Lemma (q.v.~\Caref{lem:improved-chernoff}) applied to the contraction $C(\tfrac{1}{n})=Pe^{\frac{1}{n}\cL}P$ on $P\cX$. Then, for all $x\in\cX$
					\begin{equation*}
						\norm{\left(Pe^{\frac{1}{n}\cL}P\right)^nx-e^{nP(e^{\frac{1}{n}\cL}-\1)P}Px}\leq \frac{n}{2}\norm{\left(P\left(e^{\frac{1}{n}\cL}-\1\right)P\right)^2x}.
					\end{equation*}
					Moreover, the asymptotic Zeno condition (\ref{eq:thm1-asympzeno}) and the continuity of the norm imply
					\begin{equation*}
						\norm{P^\perp\cL Px}=\lim\limits_{h\rightarrow 0}\frac{1}{h}\norm{P^\perp e^{h\cL}Px}\leq b\norm{x}
					\end{equation*}
					for all $x\in\cD(\cL P)$. Hence $P^\perp\mathcal{L}P$ is a bounded operator with $\|P^\perp \cL P\|_\infty\leq b$. Next, given $x\in\cD(\cL P)$, the $C_0$-semigroup properties (q.v.~\Caref{lem:properties-semigroups}) imply
					\begin{align*}
						n\norm{\left(P\left(e^{\frac{1}{n}\cL}-\1\right)P\right)^2x}&=\norm{P\left(e^{\frac{1}{n}\cL}-\1\right)P\int_{0}^{1}e^{\frac{\tau_1}{n}\cL}(\1-P+P)\cL Pxd\tau_1}\\
						&\leq\norm{P\left(e^{\frac{1}{n}\cL}-\1\right)P\int_{0}^{1}e^{\frac{\tau_1}{n}\cL}P\cL Pxd\tau_1}\\
						&\quad+2\int_{0}^{1}\norm{Pe^{\frac{\tau_1}{n}\cL}P^\perp}_\infty d\tau_1\norm{P^\perp\cL P}_\infty\norm{x}.
					\end{align*}
					Note that $\int_{0}^{1}e^{\frac{\tau_1}{n}\cL}P\cL Px$ belongs to $\cD(\cL)$ by \Cref{lem:properties-semigroups}, but not necessarily to $\cD(\cL P)$. However, for all $x\in\cD((\cL P)^2)$
					\begin{equation*}
						\begin{aligned}
							n\norm{\left(P\left(e^{\frac{1}{n}\cL}-\1\right)P\right)^2x}&\leq\norm{P\left(e^{\frac{1}{n}\cL}-\1\right)P\int_{0}^{1}e^{\frac{\tau_1}{n}\cL}P\cL Pxd\tau_1}+\frac{b^2}{n}\norm{x}\\
							&\leq\norm{P\left(e^{\frac{1}{n}\cL}-\1\right)\int_{0}^{1}e^{\frac{\tau_1}{n}\cL}P\cL Pxd\tau_1}\\
							&\quad+\frac{b^2}{n}\norm{x}+\norm{Pe^{\frac{1}{n}\cL}P^\perp}_\infty\norm{\cL Px}\\
							&\leq\frac{1}{n}\norm{P\int_{0}^{1}e^{\frac{\tau_2}{n}\cL}\cL\int_{0}^{1}e^{\frac{\tau_1}{n}\cL}P\cL Pxd\tau_1d\tau_2}\\
							&\quad+\frac{b^2}{n}\norm{x}+\frac{b}{n}\norm{\cL Px}\\
							&\leq\frac{1}{n}\left(b^2\norm{x}+b\norm{\cL Px}+\norm{(\cL P)^2x}\right),
						\end{aligned}
					\end{equation*}
					which proves \Cref{lem:proofthm1-term2}.
				\end{proof}
				\begin{rmk*}
					As regards to the convergence rate of the quantum Zeno effect, \Cref{lem:proofthm1-term2} constitutes our main improvement compared to the work \cites{Becker.2021}. The modified Chernoff lemma allows to improve the convergence rate to $n^{-1}$.
				\end{rmk*}
			
			\subsubsection*{Upper bound on \Cref{eq:proofthm1-term3}:}
				Finally, we prove an upper bound on \Cref{eq:proofthm1-term3}, which can be interpreted as a modified \textit{Dunford-Segal approximation} \cite{Gomilko.2014}.
				\begin{lem}\label{lem:proofthm1-term3}
					Let $(\cL,\cD(\cL))$ be the generator of a $C_0$-contraction semigroup on $\cX$ and $P\in\cB(\cX)$ be a projection. Assume that both operators satisfy the assumptions of \Cref{thm:spectral-gap}. Then, for all $x\in\cD((\cL P)^2)$
					\begin{equation*}
						\norm{e^{nP(e^{\frac{1}{n}\cL}-\1)P}Px-e^{P\cL P}Px}\leq\frac{e^{\tilde{b}}}{2n}\left(b^2\norm{x}+\norm{(\cL P)^2 x}\right)
					\end{equation*}
					with $e^{\tilde{b}}\coloneqq\sup_{s\in[0,1]}\|e^{sP\cL P}P\|_\infty<\infty$.
				\end{lem}
				\begin{proof}
					The proof relies on the integral equation for semigroups from \Cref{lem:integral-equation-semigroups}. We start by proving the continuity of 
					\begin{equation}\label{eq:proofthm1-integral-equation-continuity}
						[0,1]\ni s\mapsto -e^{sP\cL P}P\left(nP(e^{\frac{1}{n}\cL}-\1)P-P\cL P\right)e^{(1-s)nP(e^{\frac{1}{n}\cL}-\1)P}Px.
					\end{equation}
					Since for all $s\in[0,1]$ and $x\in\cD(\cL P)$
					\begin{align*}
						e^{snP(e^{\frac{1}{n}\cL}-\1)P}P\cL Px&=\lim\limits_{h\rightarrow0}Pe^{snP(e^{\frac{1}{n}\cL}-\1)P}\frac{P(e^{h\cL}-\1)P}{h}x\\
						&=\lim\limits_{h\rightarrow0}\frac{P(e^{h\cL}-\1)P}{h}e^{snP(e^{\frac{1}{n}\cL}-\1)P}Px=P\cL Pe^{snP(e^{\frac{1}{n}\cL}-\1)P}Px,
					\end{align*}
					the vector-valued function defined in \Cref{eq:proofthm1-integral-equation-continuity} is equal to
					\begin{equation*}
						-e^{sP\cL P}Pe^{(1-s)nP(e^{\frac{1}{n}\cL}-\1)P}\left(nP(e^{\frac{1}{n}\cL}-\1)P-P\cL P\right)x.
					\end{equation*}
					and, thereby, well-defined and continuous in $s$. Therefore, \Cref{lem:integral-equation-semigroups} gives for all $x\in\cD(\cL P)$
					\begin{align*}
						&e^{nP(e^{\frac{1}{n}\cL}-\1)P}Px-e^{P\cL P}Px\\
						&\qquad\quad=-\int_{0}^{1}e^{sP\cL P}Pe^{(1-s)nP(e^{\frac{1}{n}\cL}-\1)P}\left(nP(e^{\frac{1}{n}\cL}-\1)P-P\cL P\right)xds.
					\end{align*}
					Moreover, for all $x\in\cD((\cL P)^2)$
					\begin{equation}\label{eq:approx-generator}
	    				\begin{aligned}
	    				  nP(e^{\frac{1}{n}\cL}-\1)Px-P\cL Px&=P\int_{0}^{1}e^{\frac{\tau_1}{n}\cL}\cL Px d\tau_1-P\cL Px\\
	    				  &=\frac{1}{n}P\int_{0}^{1}\int_{0}^{1}\tau_1e^{\frac{\tau_1\tau_2}{n}\cL}(\cL P)^2xd\tau_2d\tau_1+P\int_{0}^{1}e^{\frac{\tau_1}{n}\cL}P^\perp\cL Px d\tau_1.
					    \end{aligned}
					\end{equation}

					Finally, we use $\sup_{s\in[0,1]}\|e^{sP\cL P}P\|_\infty<\infty$, which holds by the \textit{principle of uniform boundedness} (q.v.~proof of \Caref{prop:trotter-kato-product-formula}), the property that $(e^{s nP(e^{\frac{1}{n}\cL}-\1)P})_{s\geq0}$ is a contraction, and the upper bounds $\|P^\perp\cL P\|_\infty\leq b$ and $\|Pe^{s\cL}P^\perp\|_\infty\leq sb$ so that
					\begin{align*}
						\norm{e^{nP(e^{\frac{1}{n}\cL}-\1)P}Px-e^{P\cL P}Px}&\leq e^{\tilde{b}}\int_{0}^{1}\norm{\left(nP(e^{\frac{1}{n}\cL}-\1)P-P\cL P\right)x}ds\\
						&\leq\frac{e^{\tilde{b}}}{2n}\left(b^2\|x\|+\|(\cL P)^2 x\|\right)
					\end{align*} 
					for all $x\in\cD((\cL P)^2)$ and $e^{\tilde{b}}=\sup_{s\in[0,1]}\|e^{sP\cL P}P\|_\infty$.
				\end{proof}
				The above approximation of $e^{P\cL P}$ by $e^{nP(e^{\frac{1}{n}\cL}-\1)P}$ is similar to the Dunford-Segal approximation, which would be given by $\operatorname{exp}\big({n(\operatorname{exp}({\frac{1}{n}P\cL P})-\1)}\big)$: for the generator $(\cK,\cD(\cK))$ of a bounded $C_0$-semigroup, \citeauthor{Gomilko.2014} proved \cite[Cor.~1.4]{Gomilko.2014}
				\begin{equation*}
					\norm{e^{nt(e^{\frac{1}{n}\cK}-\1)}x-e^{t\cK}x}\leq8\tilde{b}\frac{t}{n}\norm{\cK^2x}
				\end{equation*}
				for all $x\in\cD(\cK^2)$ and $\tilde{b}\coloneqq\sup_{t\geq0}\|e^{t\cK}\|_\infty$. In our case, it is not clear whether $(e^{sP\cL P})_{s\geq0}$ is uniformly bounded.
				\begin{rmk*}
					The specificity of the last step stems from the fact that $(\cL,\cD(\cL))$ is unbounded. In the previous works \cites{Mobus.2019}{Becker.2021} a similar step exits but in both papers $\cL$ was assumed to be bounded. Moreover, \Cref{eq:approx-generator} is the only step in the proof of \Cref{thm:spectral-gap}, which deals with the operator $P\cL P$. If $P\cL P$ is closable, $P\cL Px=\overline{P\cL P}x$ for all $x\in\cD(\cL P)$ so that it is enough to ask for the closure of $P\cL P$ to define a generator. The same reasoning works for \Cref{prop:spectral-gap-uniform-norm-power-convergence} and \Cref{cor:spectral-gap}.
				\end{rmk*}
			
			\subsubsection*{End of the proof of \Cref{thm:spectral-gap}:}
				We combine \Cref{lem:proofthm1-term1}, \ref{lem:proofthm1-term2}, and \ref{lem:proofthm1-term3} to prove \Cref{thm:spectral-gap}.
				\begin{proof}[Proof of \Cref{thm:spectral-gap}]
					Let $x\in\cD((\cL P)^2)$. Then,
					\begin{flalign*}
						&&\norm{\left(Me^{\frac{1}{n}\cL}\right)^nx-e^{P\cL P}Px}&\leq\left(\delta^n+\frac{b}{n}+\frac{1}{n}\frac{b(2+b)(\delta-\delta^{n})}{1-\delta}e^{2b}\right)\norm{x}&&\text{(\Caref{lem:proofthm1-term1})}\\
						&& &+\frac{1}{2n}\left(b^2\norm{x}+b\norm{\cL Px}+\norm{(\cL P)^2x}\right)&&\text{(\Caref{lem:proofthm1-term2})}\\
						&& &+\frac{e^{\tilde{b}}}{2n}\left(b^2\norm{x}+\norm{(\cL P)^2 x}\right)&&\text{(\Caref{lem:proofthm1-term3})}.
					\end{flalign*}
					Redefining $\cL$ by $t\cL$ and $b$ by $tb$, we achieve
					\begin{equation*}
						\norm{\left(Me^{\frac{t}{n}\cL}\right)^nx-e^{tP\cL P}Px}\leq\frac{c}{n}\left(\norm{x}+\norm{\cL Px}+\norm{(\cL P)^2x}\right)+\delta^n\norm{x}
					\end{equation*}
					with an appropriate constant $c>0$ and $e^{\tilde{b}}=\sup_{s\in[0,t]}\|e^{sP\cL P}\|_\infty$.
				\end{proof}
				\begin{rmk*}
					The upper bound in \Cref{thm:spectral-gap} can be formulated for all $x\in\cD(\cL P)$. For this, one must stop at an earlier stage of the proof and express the error terms by appropriate integrals. One possible bound would be the following 
					\begin{flalign*}
						&& \|\left(Me^{\frac{1}{n}\cL}\right)^nx&-e^{P\cL P}Px\| &&\\
						&& &\leq\left(\delta^n+\frac{b}{n}+\frac{1}{n}\frac{b(2+b)(\delta-\delta^{n})}{1-\delta}e^{2b}\right)\norm{x}&&\text{(\Caref{lem:proofthm1-term1})}\\
						&& &+\frac{1}{2n}\left(b^2\norm{x}+b\norm{\cL Px}+\norm{P\int_{0}^{1}\int_{0}^{1}\cL e^{\frac{\tau_1+\tau_2}{n}\cL}P\cL Pxd\tau_1d\tau_2}\right)&&\text{(\Caref{lem:proofthm1-term2})}\\
						&& &+\frac{e^{\tilde{b}}}{2n}\left(b^2\norm{x}+2\norm{P\int_{0}^{1}\int_{0}^{1}\tau_1\cL e^{\frac{\tau_1\tau_2}{n}\cL}P\cL Pxd\tau_2d\tau_1}\right)&&\text{(\Caref{lem:proofthm1-term3})}.
					\end{flalign*}
				\end{rmk*}
			\subsection{Proof of \texorpdfstring{\Cref{cor:explicit-bound-thm1} and \Cref{cor:explicit-bound-thm1-closed-sys}}{Proposition 3.1 and Corollary 3.2}}\label{subsec:proofexplicitboundthm1}
				\begin{proof}[Proof of \Cref{cor:explicit-bound-thm1}]
					Since $\|P\|_\infty\leq1$ (\ref{eq:projection-contraction}) and $t\mapsto e^{t\cL}$ is a uniformly continuous contraction semigroup, the generator is defined on $\cX$ and bounded, i.e.~$\|\cL\|_\infty<\infty$, so that
					\begin{equation*}
						\norm{Pe^{t\cL}(\1-P)}_\infty=\norm{P(e^{t\cL}-\1)(\1-P)}_\infty\leq t\norm{\cL\int_0^1e^{ts\cL}(\1-P)ds}_\infty\leq t\norm{\cL}_\infty c_{p}.
					\end{equation*}
					where $c_{p}\coloneqq\|\1-P\|_\infty\leq2$. Then, we simplify the bounds found in \Cref{lem:proofthm1-term1}, \ref{lem:proofthm1-term2}, and \ref{lem:proofthm1-term3} to
					\begin{equation*}
						\norm{\left(Me^{\frac{1}{n}\cL}\right)^n-e^{P\cL P}P}_\infty\leq\delta^n+\frac{1}{n}\left(b+\frac{2\delta}{1-\delta}e^{3b}+\frac{b^2}{2}+\frac{b}{2}\norm{\cL}_\infty+\frac{1}{2}\norm{\cL}^2_\infty+\frac{e^{\tilde{b}}b^2}{2}+\frac{e^{\tilde{b}}}{2}\norm{\cL}^2_\infty\right),
					\end{equation*}
					where $e^{\tilde{b}}=\sup_{s\in[0,t]}\|e^{sP\cL P}\|_\infty$. By redefining $\cL$ by $t\cL$, $b$ by $tb$, and using $b\leq\|\cL\|_\infty c_{p}$,
					\begin{align*}
						\norm{\left(Me^{\frac{t}{n}\cL}\right)^n-e^{tP\cL P}P}_\infty\leq c_{p}\frac{t\|\cL\|_\infty}{n} &+\left(\frac{c_{p}+(1+e^{\tilde{b}})(1+c_{p}^2)}{2}\right)\frac{t^2\|\cL\|_\infty^2}{n}\\
						&+\delta^n+\frac{2\delta}{1-\delta}\frac{e^{3t\|\cL\|_\infty c_{p}}}{n}
					\end{align*}
					which proves the statement.
				\end{proof}
				\begin{proof}[Proof of \Cref{cor:explicit-bound-thm1-closed-sys}]
					In closed quantum systems $\cX=\cH$ equipped with the operator norm induced by the scalar product, which shows $\|U\|_\infty=1$ for all unitaries $U\in\cB(\cH)$. Especially, $e^{\tilde{b}}=\sup_{s\in[0,t]}\|e^{sP\cL P}\|_\infty=1$ because $\|P\|_\infty=1$ is equivalent to $P=P^\dagger$ \cite[Thm.~2.1.9]{Simon.2015} so that $PHP$ is hermitian. Moreover, $P=P^\dagger$ implies $(\1-P)^\dagger=(\1-P)$ which shows $c_{p}=\|\1-P\|_\infty\leq1$. Finally, the choice $M=P$ implies $\delta=0$ which proves the \nameCref{cor:explicit-bound-thm1-closed-sys} by inserting the constants into \Cref{cor:explicit-bound-thm1}.
				\end{proof}
			\subsection{Proof of \texorpdfstring{\Cref{prop:spectral-gap-uniform-norm-power-convergence} }{Proposition 5.7}}\label{subsec:thm1-prop}~\\
				In this subsection, we weaken the assumptions (\ref{unifpower1}) on the contraction $M$ at the cost of stronger assumptions on the $C_0$-semigroup. For that, we combine techniques from holomorphic functional calculus with the semicontinuity of the spectrum of $M$ perturbed by the semigroup under certain conditions. We refer to \Cref{sec:appendix-holomorphic-fc-semicontinuity} for details on the tools needed to prove the main result of this section.
				\begin{prop}\label{prop:spectral-gap-uniform-norm-power-convergence}
			    	Let $(\cL,\cD(\cL))$ be the generator of a $C_0$-contraction semigroup on $\cX$ and $M\in\cB(\cX)$ a contraction such that
			    	\begin{equation}\label{unifpowerc}
			    		\norm{M^n-P}_\infty\leq\tilde{c}\,\delta^n
		    		\end{equation}
	    			for some projection $P$, $\delta\in(0,1)$ and $\tilde{c}\ge 0$. Moreover, we assume that there is $b\geq0$ so that 
	    			\begin{equation}\label{key}
	    				\norm{Pe^{t\cL}P^\perp}_\infty\leq tb,\quad\norm{M^\perp e^{t\cL}-M^\perp}_\infty\leq tb,\quad\text{and}\quad\norm{P^\perp e^{t\cL}P}_\infty\leq tb
    				\end{equation}
    				where $M^\perp=(\1-P)M$. If $(P\cL P,\cD(\cL P))$ generates a $C_0$-semigroup, then there is $\epsilon>0$ such that for all $t\geq0$, $n\in\N$ satisfying $t\in[0,n\epsilon]$, $\tilde{\delta}\in(\delta,1)$, and $x\in\cD((\cL P)^2)$ 
					\begin{equation*}
						\norm{\left(Me^{\frac{t}{n}\cL}\right)^nx-e^{tP\cL P}Px}\leq\frac{c_1(t,b,\tilde{\delta}-\delta)}{n}\left(\norm{x}+\norm{\cL P x}+\norm{(\cL P)^2x}\right)+c_2(\tilde{c},\tilde{\delta}-\delta)\tilde{\delta}^n\norm{x}\,,
					\end{equation*}
					for some constants $c_1,c_2\ge 0$ depending on $t$, $b$, the difference $\tilde{\delta}-\delta$, and $\tilde{c}$.
				\end{prop}
				The only difference to the proof of \Cref{thm:spectral-gap} is summarized in the question: How can we upper bound $\|(M^\perp e^{\frac{t}{n}\cL})^k\|_\infty$ for all $k\in\{1,...,n\}$ with the weaker assumption (\ref{unifpowerc}) on $M$? For that, we replace the argument in the proof of \Cref{lem:proofthm1-term1}, which only works for the case $\tilde{c}=1$.
				\begin{proof}
					Since the bounds found in \Cref{lem:proofthm1-term2} and \ref{lem:proofthm1-term3} are independent of the value of $\tilde{c}$, it is enough to improve \Cref{lem:proofthm1-term1}:
					\begin{flalign*}
						&&\norm{\left(Me^{\frac{t}{n}\cL}\right)^nx-e^{tP\cL P}Px}&\leq\norm{\left(Me^{\frac{t}{n}\cL}\right)^nx-\left(Pe^{\frac{t}{n}\cL}P\right)^nx}&&\\
						&& &+\frac{t^2}{2n}\left(b^2\norm{x}+b\norm{\cL Px}+\norm{(\cL P)^2x}\right)&&\text{(\Caref{lem:proofthm1-term2})}\\
						&& &+\frac{e^{\tilde{b}}t^2}{2n}\left(b^2\norm{x}+\norm{(\cL P)^2 x}\right)&&\text{(\Caref{lem:proofthm1-term3})}
					\end{flalign*}
					with $e^{\tilde{b}}=\sup_{s\in[0,t]}\|e^{sP\cL P}\|_\infty<\infty$. Since the assumption (\ref{unifpowerc}) on $M$ is a special case of the uniform power convergence (q.v.~\Caref{unifpower}), \Cref{prop:equivalence-spectral-gap} shows the equivalence of the uniform\\[1ex]
					\begin{minipage}[c]{0.69\textwidth}
						power convergence of $M$ to the spectral gap assumption, that is
						\begin{equation*}
							\sigma(M)\subset\D_\delta\cup\{1\},
						\end{equation*}
						where the quasinilpotent operator corresponding to the eigenvalue $1$ vanishes. Therefore, the eigenprojection $P$ w.r.t.~$1$ satisfies $MP=PM=P$ and the curve 	$\gamma:[0,2\pi]\rightarrow\C,\varphi\mapsto\tilde{\delta}e^{i\varphi}$ encloses the spectrum of $M^\perp\coloneqq MP^\perp$ (q.v.~\Caref{fig3}). Together with the second bound in \eqref{key}, \Cref{lem:semicontinuity-spectrum} shows that there exists $\epsilon>0$ so that the spectrum of $M^\perp e^{s\cL}$ can be separated by $\gamma$ for all $s\in[0,\epsilon]$.
					\end{minipage}
					\begin{minipage}[c]{0.3\textwidth}
						\begin{center}
							\begin{tikzpicture}
								\draw[dashed,tumivory] (0,0) ellipse (1.4cm and 1.4cm);
								\draw[dashed,mygreen] (0,0) ellipse (1.15cm and 1.15cm);
								\filldraw (1.4,0) ellipse (0.03cm and 0.03cm);
								\filldraw (0,0) ellipse (0.03cm and 0.03cm);
								\draw (1.7,0) node {$1$};
							
								%\draw[<->] (-0.025,0.0433)--(-0.475cm,0.8227cm);
								\draw[<->] (0.0,0.05)--(0,0.95);
								\draw[<->,mygreen] (0.0433,-0.025)--(0.9526cm,-0.55cm);
							
								\fill[pattern=my north east lines] (0,0) ellipse (1cm and 1cm);
								\filldraw[white] (-0.35,0.5) ellipse (0.15cm and 0.25cm);
								\draw (-0.35,0.5) node {$\delta$};
								\filldraw[white] (0.35,-0.5) ellipse (0.15cm and 0.25cm);
								\draw[mygreen] (0.35,-0.5) node {$\tilde{\delta}$};
								\filldraw[white] (-1.06,-1.06) ellipse (0.15cm and 0.3cm);
								\draw (-1.06,-1.06) node [mygreen]{$\gamma$};
							
								\draw (1.5,1.2) node {$\sigma(M)$};
								%\draw (0,0) node [fill=white]{$\spec(B)\backslash\{1\}$};
								
								%\draw (1.9,1.3) node {$\mathbb{C}$};
							\end{tikzpicture}\\[1ex]
							\captionof{figure}{}\label{fig3}
						\end{center}
					\end{minipage}
					Therefore, the holomorphic functional calculus (q.v.~\Caref{prop:holomorphic-functional-calculus}) shows for all $t\in[0,n\epsilon], k\in\N$
					\begin{equation}\label{eq:proofthm1-upper-bound-inner-part}
						\left(M^\perp e^{\frac{t}{n}\cL}\right)^{k}=\frac{1}{2\pi i}\oint_{\gamma}z^{k}R(z,M^\perp e^{\frac{t}{n}})dz.
					\end{equation}
					Let $t\geq0$, $n\in\N$ so that $t\in[0,n\epsilon]$. By the \textit{principle of stability of bounded invertibility} \cite[Thm.~IV.2.21]{Kato.1995}, $R(z,M^\perp e^{s\cL})$ is well-defined and bounded for all $z\in\gamma$ and $s\in[0,\epsilon]$. More explicitly, using the \textit{second Neumann series} for the resolvent \cite[p.~67]{Kato.1995}, we have 
					\begin{equation}\label{eq:proofthm1-perturbed-resolvent-uniform-boundedness}
						\begin{aligned}
							\norm{R(z,M^\perp e^{s\cL})}_\infty&=\norm{R(z,M^{\perp})\sum_{p=0}^{\infty}\left((M^\perp e^{s\cL}-M^\perp)R(z,M^\perp)\right)^p}_\infty\\
							&\leq\norm{R(z,M^\perp)}_\infty\sum_{p=0}^{\infty}\left(sb\norm{R(z,M^\perp)}_\infty\right)^p\\
							&\leq\sup_{z\in\gamma}\norm{R(z,M^\perp)}_\infty \frac{2+2\tilde{\delta}^2}{1+2\tilde{\delta}^2}\eqqcolon c_2,
						\end{aligned}
					\end{equation}
					where we have applied the assumption \eqref{key} and the following upper bound on $s$ (q.v.~\Caref{eq:semicontinuity-time-bound}):
					\begin{equation*}
						s\leq\epsilon<\frac{1}{2b}(1+\tilde{\delta}^2)^{-1}\left(1+\sup_{z\in\gamma}\norm{R(z,M^\perp)}_\infty^2\right)^{-\frac{1}{2}}\leq\frac{1}{2b}(1+\tilde{\delta}^2)^{-1}\left(\sup_{z\in\gamma}\norm{R(z,M^\perp)}_\infty\right)^{-1},
					\end{equation*}
					to compute the geometric series. Combining \Cref{eq:proofthm1-upper-bound-inner-part} and (\ref{eq:proofthm1-perturbed-resolvent-uniform-boundedness}) shows for all $k\in\{1,...,n\}$
					\begin{equation}\label{eq:proofprop-perturbed-measurement}
						\norm{\left(M^\perp e^{\frac{t}{n}\cL}\right)^{k-1}}_\infty\leq\frac{1}{2\pi}\oint_{\gamma}\abs{z}^{k-1}\norm{R(z,M^\perp e^{\frac{t}{n}\cL})}_\infty dz\leq c_2\tilde{\delta}^{k}.
					\end{equation}
					Next, we define $A\coloneqq M^\perp e^{\frac{{t}}{n}\cL}$ and $B\coloneqq Pe^{\frac{{t}}{n}\cL}$ and expand
					\begin{equation*}
						\left(Me^{\frac{t}{n}\cL}\right)^n=\left(PMe^{\frac{t}{n}\cL}+M^\perp e^{\frac{t}{n}\cL}\right)^n=(B+A)^n.
					\end{equation*}
					The above $n^{\text{th}}$-power can be expanded in terms of sequences of the form 
					\begin{equation*}
						A...AB...BA...AB...\quad\text{or}\quad B...BA...AB...BA...
					\end{equation*}
					Similarly to \Cref{lem:proofthm1-term1}, we can upper bound every sequence w.r.t.~the number of transitions $AB$ or $BA$ using the assumptions \eqref{key} on the $C_0$-semigroup as well as the inequality (\ref{eq:proofprop-perturbed-measurement}). The only difference to the proof of \Cref{lem:proofthm1-term1} is the constant $c_2$ in the inequality so that
					\begin{equation*}
						\begin{aligned}
							\norm{\left(Me^{\frac{t}{n}\cL}\right)^n-\left(Pe^{\frac{t}{n}\cL}\right)^n}_\infty&\leq c_2\tilde{\delta}^n+\sum_{k=1}^{n-1}\sum_{j=1}^{m}\tilde{\delta}^{k}N(j,n,k)\left(\frac{tbc_2}{n}\right)^j\\
							&\leq c_2\tilde{\delta}^n+\frac{b}{n}+\frac{1}{n}\frac{bc_2(2+bc_2)(\tilde{\delta}-\tilde{\delta}^{n})}{1-\tilde{\delta}}e^{2bc_2},
						\end{aligned}
					\end{equation*}
					where $m\coloneqq2\min\{k,n-k\}-1_{2k=n}$. Then, for all $x\in\cD((\cL P)^2)$ and an appropriate $c_1\geq0$\\
					\begin{minipage}[b]{0.9\textwidth}
						\begin{equation*}
							\norm{\left(Me^{\frac{t}{n}\cL}\right)^nx-e^{tP\cL P}Px}\leq\frac{c_1}{n}\left(\norm{x}+\norm{\cL Px}+\norm{(\cL P)^2x}\right)+c_2{\tilde{\delta}}^n\norm{x}.
						\end{equation*}
					\end{minipage}
				\end{proof}
				\vspace{-2ex}
				\begin{proof}[Proof of \Cref{cor:explicit-bound-prop}]
					Similarly to \Cref{cor:explicit-bound-thm1},
					\begin{equation*}
						\norm{\left(Me^{\frac{t}{n}\cL}\right)^n-e^{tP\cL P}P}_\infty\leq c_2\tilde{\delta}^n+\frac{tc_{p}\|\cL\|_\infty}{n} +\frac{c_{p}+(1+e^{\tilde{b}})(1+c_{p}^2)}{2}\frac{t^2\|\cL\|_\infty^2}{n}+\frac{2\tilde{\delta}}{1-\tilde{\delta}}\frac{e^{3tc_{p}c_2\|\cL\|_\infty}}{n}
					\end{equation*}
					where $e^{\tilde{b}}=\sup_{s\in[0,t]}\|e^{sP\cL P}\|_\infty$ and $c_2\coloneqq\sup_{z\in\gamma}\norm{R(z,M^\perp)}_\infty \frac{2+2\tilde{\delta}^2}{1+2\tilde{\delta}^2}$ (see \Caref{eq:proofthm1-perturbed-resolvent-uniform-boundedness}).
					The constant $c_2$ can be bounded with the help of the \textit{first von Neumann series} \cite[p.~37]{Kato.1995} and the geometric series:
					\begin{equation*}
						\frac{2+2\tilde{\delta}^2}{1+2\tilde{\delta}^2}\sup\limits_{z\in\gamma}\norm{R(z,M^\perp)}_\infty=2\sup\limits_{z\in\gamma}\|\sum_{k=0}^{\infty}z^{-(k+1)}(M^\perp)^k\|_\infty\leq2\frac{\tilde{c}}{\tilde{\delta}}\sum_{k=0}^{\infty}\tilde{\delta}^{-k}\delta^k=\frac{2\tilde{c}}{\tilde{\delta}-\delta}
					\end{equation*}
					so that 
					\begin{align*}
	    	            \norm{\left(Me^{\frac{t}{n}\cL}\right)^n-e^{tP\cL P}P}_\infty\leq\frac{tc_{p}\|\cL\|_\infty}{n}&+\frac{c_{p}+(1+e^{\tilde{b}})(1+c_{p}^2)}{2}\frac{t^2\|\cL\|_\infty^2}{n}\\
	        	        &+\frac{2\tilde{c}}{\tilde{\delta}-\delta}\tilde{\delta}^n+\frac{2\tilde{\delta}}{1-\tilde{\delta}}\frac{e^{\frac{6tc_{p}\tilde{c}\|\cL\|_\infty}{\tilde{\delta}-\delta}}}{n}
	            	\end{align*}
            		which finishes the proof of the \namecref{cor:explicit-bound-prop}.
				\end{proof}
		
	\section{Uniform Power Convergence with Finitely Many Eigenvalues}\label{sec:finitely-many-eigenvalues}
		In this section, we weaken the assumption on $M$ to the uniform power convergence assumption (q.v.~\Caref{unifpower}), that is we allow for finitely many eigenvalues $\{\lambda_j\}_{j=1}^J$ and associated projections $\{P_j\}_{j=1}^J$ satisfying $P_jP_k=1_{j=k}P_j$. Similarly to Theorem 3 in \cite{Becker.2021}, we strengthen the assumptions on the $C_0$-semigroup to $M\cL$ and $\cL P_\Sigma$ being densely defined and bounded by $b\geq0$, where $P_\Sigma\coloneqq \sum_{j=1}^{J}P_j$. Under those assumptions, we can prove the Zeno convergence in the uniform topology:

		\begin{thm}\label{thm:spectral-gap-uniform}
		    Let $(\cL,\cD(\cL))$ be the generator of a $C_0$-contraction semigroup on $\cX$ and $M\in\cB(\cX)$ a contraction satisfying the following uniform power convergence: there is $\tilde{c}>0$ so that
			\begin{equation}\label{unifpower}
			    \norm{M^n-\sum_{j=1}^{J}\lambda^n_jP_j}_\infty\leq\tilde{c}\,\delta^n
		    \end{equation}
		    for a set of projections $\{P_j\}_{j=1}^J$ satisfying $P_jP_k=1_{j=k}P_j$, eigenvalues $\{\lambda_j\}_{j=1}^J\subset\partial\D_1$, and a rate $\delta\in(0,1)$. For $P_\Sigma\coloneqq\sum_{j=1}^{J}P_j$, we assume that $M\cL$ and $\cL P_\Sigma$ are densely defined and bounded by $b\geq0$ and $\|P_j\|_\infty=1$ for all $j\in\{1,...,J,\Sigma\}$. Then, there is an $\epsilon>0$ such that for all $n\in\N$, $t\geq0$ satisfying $t\in[0,n\epsilon]$, and $\tilde{\delta}\in(\delta,1)$
			\begin{equation*}
				\norm{\left(Me^{\frac{t}{n}\cL}\right)^n-\sum_{j=1}^{J}\lambda_j^ne^{tP_j\cL P_j}P_j}_\infty\leq\frac{c_1 }{n}+c_2\tilde{\delta}^n\,,
			\end{equation*}
			for some constants $c_1,c_2\ge 0$ depending on all involved parameters except from $n$.
		\end{thm}
		\subsection{Proof of \texorpdfstring{\Cref{thm:spectral-gap-uniform} }{Theorem 3.3 }}~\\
		Similarly to the papers \cite{Mobus.2019} and \cite{Becker.2021}, we use the holomorphic functional calculus to
		separate the spectrum of the contraction $Me^{t\cL}$ appearing in the Zeno sequence. In contrast to \cite{Becker.2021} where the $C_0$-semigroup is approximated by a sequence of uniformly continuous semigroups, we instead crucially rely upon the uniform continuity of the perturbed contraction to recover the optimal convergence rate. We upper bound the following terms: 
		\begin{align}
			\norm{\left(Me^{\frac{t}{n}\cL}\right)^n-\sum_{j=1}^{J}\lambda_j^ne^{tP_j\cL P_j}P_j}_\infty&\leq\norm{\left(Me^{\frac{t}{n}\cL}\right)^n-\left(P_\Sigma Me^{\frac{t}{n}\cL}P_\Sigma\right)^n}_\infty\label{eq:proofthm2-term1}\\
			&+\norm{\left(P_\Sigma Me^{\frac{t}{n}\cL}P_\Sigma\right)^n-\sum_{j=1}^{J}\lambda_j^ne^{n\left(C_{j}(\text{\tiny{$\tfrac{t}{n}$}})-P_{j}(\text{\tiny{$\tfrac{t}{n}$}})\right)}P_{j}(\tfrac{t}{n})}_\infty\label{eq:proofthm2-term2}\\
			&+\norm{\sum_{j=1}^{J}\lambda_j^ne^{n\left(C_{j}(\text{\tiny{$\tfrac{t}{n}$}})-P_{j}(\text{\tiny{$\tfrac{t}{n}$}})\right)}P_{j}(\tfrac{t}{n})-\sum_{j=1}^{J}\lambda_j^ne^{tP_j\cL P_j}P_j}_\infty\label{eq:proofthm2-term3}
		\end{align}
		where the definitions of the \textit{perturbed spectral projection} $P_j(\tfrac{t}{n})$ and the \textit{Chernoff contraction} $C_{j}(\tfrac{t}{n})$ are postponed to \Cref{lem:proofthm2-term2}.
		
		\subsubsection*{Approximation of the Perturbed Spectral Projection:}
			 In the following result, we consider an operator $A$ uniformly perturbed by a vector-valued map $t\mapsto B(t)$ in the following way:
			 \begin{equation*}
			     t\mapsto A+tB(t).
			 \end{equation*}
			 Under certain assumptions on the perturbation controlled by $t$, we construct the associated perturbed spectral projection for which we obtain an approximation bound (cf.~\cite[Lem.~5.3]{Becker.2021}). The key tools are the holomorphic functional calculus and the semicontinuity of the spectrum under \textit{uniform perturbations}. The statements are summarized in \Cref{prop:holomorphic-functional-calculus} and \Cref{lem:semicontinuity-spectrum}.
			\begin{lem}\label{lem:quantitative-appro-riesz-projection}
				Let $A\in\cB(\cX)$, $t\mapsto B(t)$ be  a vector-valued map on $\cB(\cX)$ which is uniformly continuous at $t=0$ with $\sup_{t\geq0}\|B(t)\|_\infty\leq b$, and $\Gamma:[0,2\pi]\rightarrow\rho(A)$ be a curve separating $\sigma(A)$. Then, there exists an $\epsilon>0$ so that for all $t\in[0,\epsilon]$
				\begin{equation*}
					P(t)\coloneqq \frac{1}{2\pi i} \oint_{\Gamma}R(z,A+tB(t))dz
				\end{equation*}
				defines a projection with $\norm{P(t)}_\infty\leq\tfrac{d|\Gamma|}{2\pi}$ and derivative at $t=0$ given by
				\begin{equation}\label{eq:thm2-term1-derivative}
					P'\coloneqq\lim\limits_{t\rightarrow0}\frac{P(t)-{P(0)}}{t}=\frac{1}{2\pi i}\oint_{\Gamma}R(z,A)B(0)R(z,A)dz.
				\end{equation}
				with $\|P'\|_\infty\leq\frac{R^2b|\Gamma|}{2\pi}$. The zeroth order approximation of $t\mapsto P(t)$ can be controlled by
				\begin{equation}\label{eq:zeroth-order-approximation}
				    \|P(t)-P\|_\infty\leq \frac{tRbd|\Gamma|}{2\pi}
				\end{equation}
				and the first order approximation by 
				\begin{equation}\label{eq:first-order-approximation}
					\norm{{P(t)-P-tP'}}_\infty\leq\frac{tR^2|\Gamma|}{2\pi}\left(tb^2d+\norm{B(t)-B(0)}_\infty\right).
				\end{equation}
				Above $|\Gamma|$ denotes the length of the curve $\Gamma$, $P$ abbreviates the unperturbed spectral projection $P(0)$, $R\coloneqq \sup_{z\in\Gamma}\|R(z,A)\|_\infty<\infty$, and $d=R\inf_{z\in\Gamma}\frac{2+2|z|^2}{1+2|z|^2}$.
			\end{lem}
		\begin{proof}
			Since $t\mapsto B(t)$ is uniformly continuous at $t=0$, the vector-valued map $t\mapsto A+tB(t)$ is uniformly continuous as well. Then, \Cref{lem:semicontinuity-spectrum} states that there exists an $\epsilon>0$ such that $\sigma(A+tB(t))$ is separated by $\Gamma$ for all $t\in[0,\epsilon]$ and \Cref{prop:holomorphic-functional-calculus} shows that 
			\begin{equation*}
				P(t)=\frac{1}{2\pi i}\oint_{\Gamma}R(z,A+tB(t))dz
			\end{equation*}
			defines a projection on $\cX$ for all $t\in[0,\epsilon]$. Let $R\coloneqq \sup_{z\in\Gamma}\|R(z,A)\|_\infty<\infty$, then using the same steps as in \Cref{eq:proofthm1-perturbed-resolvent-uniform-boundedness}, we have that for all $\eta\in\Gamma$
			\begin{equation}\label{eq:proofthm2-perturbed-resolvent-uniform-boundedness}
				\norm{R(\eta,A+tB(t))}_\infty\leq R\inf_{z\in\Gamma}\frac{2+2|z|^2}{1+2|z|^2}\eqqcolon d\,.
			\end{equation}
	        Therefore, the perturbed resolvent is uniformly bounded. To prove the explicit representation of the derivative and the quantitative approximation, we follow the ideas of \cite[Lem.~5.2-5.3]{Becker.2021}: 
			\begin{align*}
				\frac{P(t)-P(0)}{t}&=\frac{1}{t2\pi i}\left(\oint_{\Gamma}R(z,A+tB(t))dz-\oint_{\Gamma}R(z,A)dz\right)\\
				&=\frac{1}{2\pi i}\oint_{\Gamma}R(z,A+tB(t))B(t)R(z,A)dz,
			\end{align*}
			which uses the \textit{second resolvent identity}, i.e.~$R(z,A+tB(t))tB(t)R(z,A)=R(z,A)-R(z,A+tB(t))$ for all $z\in\Gamma$ and $t\in[0,\epsilon]$ and proves \Cref{eq:thm2-term1-derivative} by Lebesgue's dominated convergence theorem \cite[Thm.~3.7.9]{Hille.2000}. Moreover, the above equation proves \Cref{eq:zeroth-order-approximation}. Finally,
			\begin{equation}\label{eq:prooflemapprox-constant}
				\begin{aligned}
					\norm{P(t)-P-tP'}_\infty
					&\leq \frac{t}{2\pi}\norm{\oint_{\Gamma}(R(z,A+tB(t))-R(z,A))B(t)R(z,A)dz}_\infty\\
					&\quad+\frac{t}{2\pi}\norm{\oint_{\Gamma}R(z,A)\left(B(t)-B(0))\right)R(z,A)dz}_\infty\\
					&\leq \frac{t}{2\pi}\norm{\oint_{\Gamma}R(z,A+tB(t))tB(t)R(z,A)B(t)R(z,A)dz}_\infty\\
					&\quad+\frac{t}{2\pi}\norm{\oint_{\Gamma}R(z,A)\left(B(t)-B(0))\right)R(z,A)dz}_\infty\\
					&\leq \frac{tR^2|\Gamma|}{2\pi}\left(tb^2d+\norm{B(t)-B(0)}_\infty\right)\\
					%&\leq c_1t^2+td\norm{B(t)-B(0)}_\infty
				\end{aligned}
			\end{equation}
			where $|\Gamma|$ denotes the length of the curve $\Gamma$.
		\end{proof}
		Now, we are ready to prove \Cref{thm:spectral-gap-uniform}.
		
		\subsubsection*{Upper bound on \Cref{eq:proofthm2-term1}:}
		\begin{lem}\label{lem:proofthm2-term1} 
			Let $(\cL,\cD(\cL))$ be the generator of a $C_0$-contraction semigroup on $\cX$ and $M\in\cB(\cX)$ a contraction with the same assumption as in \Cref{thm:spectral-gap-uniform} and $c_p\coloneqq\|\1-P_\Sigma\|_\infty$. Then, there is an $\epsilon_1>0$ and $c_2\geq0$ so that for all $t\geq0$ and $n\in\N$ satisfying $t\in[0,n\epsilon_1]$
			\begin{equation*}
				\norm{\left(Me^{\frac{t}{n}\cL}\right)^n-\left(P_\Sigma Me^{\frac{t}{n}\cL}P_\Sigma\right)^n}\leq c_2\tilde{\delta}^n+\frac{tb}{n}+\frac{1}{n}\frac{tbc_{p}c_2(2+tbc_{p}c_2)(\tilde{\delta}-\tilde{\delta}^{n})}{1-\tilde{\delta}}e^{2tbc_{p}c_2}\,.
			\end{equation*}
		\end{lem}

		\begin{proof}
			As in the proof of \Cref{prop:spectral-gap-uniform-norm-power-convergence}, \Cref{prop:equivalence-spectral-gap} shows that the uniform power conver-\linebreak
			\begin{minipage}[c]{0.69\textwidth}
				gence assumption (\ref{unifpower}), that is $\|M^n-\sum_{j=1}^J\lambda_j^{{ n}}P_j\|_\infty\leq\tilde{c}\,\delta^n$ for all $n\in\N$, is equivalent to the spectral gap assumption (q.v. \Caref{sec:appendix-spectral-gap}),
				\begin{equation*}
					\sigma(M)\subset\D_\delta\cup\{\lambda_1,...,\lambda_J\},
				\end{equation*}
				with corresponding quasinilpotent operators being zero. Therefore, the curve $\gamma:[0,2\pi]\rightarrow\C,\varphi\mapsto\tilde{\delta}e^{i\varphi}$, with $\tilde{\delta}\in (\delta,1)$, encloses the spectrum of $M^\perp\coloneqq MP_\Sigma^\perp=P_\Sigma^\perp M$, where $P_\Sigma=\sum_{j=1}^JP_j$ and $P_\Sigma^\perp=\1-P_\Sigma$ (q.v.~\Caref{fig1}). By \Cref{lem:properties-semigroups}
				\begin{equation*}
					\norm{M^\perp e^{s\cL}-M^\perp}_\infty=s\norm{M^\perp\cL\int_{0}^{1}e^{\tau s\cL}d\tau}_\infty\leq sc_{p}b
				\end{equation*}
			    with $c_{p}\coloneqq\|P_{\Sigma}^\perp\|_\infty$. Therefore, $M^\perp e^{s\cL}$ converges uniformly to $M^\perp$ for $s\downarrow0$. Hence, \Cref{lem:semicontinuity-spectrum} shows that there exists an \linebreak\vspace{-2ex}
			\end{minipage}
			\begin{minipage}[c]{0.3\textwidth}
				\begin{center}
					\begin{tikzpicture}
						\draw[dashed,tumivory] (0,0) ellipse (1.4cm and 1.4cm);
						\draw[dashed,mygreen] (0,0) ellipse (1.15cm and 1.15cm);
						\filldraw (0,0) ellipse (0.03cm and 0.03cm);
						
						\filldraw (15:1.4) ellipse (0.03cm and 0.03cm);
						%\draw[dashed,myorange] (15:1.4) ellipse (0.2cm and 0.2cm);
						\draw (15:1.7) node {$\lambda_1$};
						
						\filldraw (70:1.4) ellipse (0.03cm and 0.03cm);
						%\draw[dashed,myorange] (70:1.4) ellipse (0.2cm and 0.2cm);
						\draw (70:1.7) node {$\lambda_2$};
						
						\filldraw (110:1.4) ellipse (0.03cm and 0.03cm);
						%\draw[dashed,myorange] (110:1.4) ellipse (0.2cm and 0.2cm);
						\draw (110:1.7) node {$\lambda_3$};
					
						\filldraw (130:1.4) ellipse (0.03cm and 0.03cm);
						%\draw[dashed,myorange] (130:1.4) ellipse (0.2cm and 0.2cm);
						\draw (130:1.7) node {$\lambda_4$};
						
						\filldraw (170:1.4) ellipse (0.03cm and 0.03cm);
						%\draw[dashed,myorange] (170:1.4) ellipse (0.2cm and 0.2cm);
						\draw (170:1.7) node {$\lambda_5$};
						
						\filldraw (195:1.4) ellipse (0.03cm and 0.03cm);
						%\draw[dashed,myorange] (195:1.4) ellipse (0.2cm and 0.2cm);
						\draw (195:1.7) node {$\lambda_6$};
						
						\filldraw (260:1.4) ellipse (0.03cm and 0.03cm);
						%\draw[dashed,myorange] (260:1.4) ellipse (0.2cm and 0.2cm);
						\draw (260:1.7) node {$\lambda_7$};
						
						\filldraw (300:1.4) ellipse (0.03cm and 0.03cm);
						%\draw[dashed,myorange] (300:1.4) ellipse (0.2cm and 0.2cm);
						\draw (300:1.7) node {$\lambda_8$};
						
						\filldraw (350:1.4) ellipse (0.03cm and 0.03cm);
						%\draw[dashed,myorange] (350:1.4) ellipse (0.2cm and 0.2cm);
						\draw (350:1.7) node {$\lambda_J$};
						
						\draw [dotted,domain=310:340] plot ({1.7*cos(\x)}, {1.7*sin(\x)});
						
						%\draw[<->] (-0.025,0.0433)--(-0.475cm,0.8227cm);
						\draw[<->] (0.0,0.05)--(0,0.95);
						\draw[<->,mygreen] (0.0433,-0.025)--(0.9526cm,-0.55cm);
						
						\fill[pattern=my north east lines] (0,0) ellipse (1cm and 1cm);
						\filldraw[white] (-0.35,0.5) ellipse (0.15cm and 0.25cm);
						\draw (-0.35,0.5) node {$\delta$};
						\filldraw[white] (0.35,-0.5) ellipse (0.15cm and 0.25cm);
						\draw[mygreen] (0.35,-0.5) node {$\tilde{\delta}$};
						\filldraw[white] (-1.06,-1.06) ellipse (0.15cm and 0.3cm);
						\draw (-1.06,-1.06) node [mygreen]{$\gamma$};
						
						\draw (1.5,1.2) node {$\sigma(M)$};
						%\draw (0,0) node [fill=white]{$\spec(B)\backslash\{1\}$};
						
						%\draw (1.9,1.3) node {$\mathbb{C}$};
					\end{tikzpicture}
					\captionof{figure}{}\label{fig1}
				\end{center}
			\end{minipage}\\
			$\epsilon_1>0$ such that the spectrum of $M^\perp e^{s\cL}$ can be separated by $\gamma$ for all $s\in[0,\epsilon_1]$. Therefore, we can apply the holomorphic functional calculus (\Caref{prop:holomorphic-functional-calculus}) to conclude that for all $t\in[0,n\epsilon_1]$
			\begin{equation*}
				\left(M^\perp e^{\frac{t}{n}\cL}\right)^{k}=\frac{1}{2\pi i}\oint_{\gamma}z^{k}R(z,M^\perp e^{\frac{t}{n}})dz,
			\end{equation*}
			where $k\in\{1,..,n\}$. By \Cref{eq:proofthm1-perturbed-resolvent-uniform-boundedness} and with $c_2\coloneqq\sup_{z\in\gamma}\|R(z,M^\perp)\|_\infty \frac{2+2\tilde{\delta}^2}{1+2\tilde{\delta}^2}$,
			\begin{equation*}
				\norm{\left(M^\perp e^{\frac{t}{n}\cL}\right)^{k-1}}_\infty\leq\frac{1}{2\pi}\oint_{\gamma}\abs{z}^{k-1}\norm{R(z,M^\perp e^{\frac{t}{n}})}_\infty dz\leq c_2\tilde{\delta}^k.
			\end{equation*}
			Moreover, by the assumptions $\|P_\Sigma\|_\infty=1$, $\|M\cL\|_\infty\leq b$, $\|\cL P_\Sigma\|_\infty\leq b$, and \Cref{lem:properties-semigroups}
			\begin{equation*}
			    \norm{MP_\Sigma e^{t\cL}P_\Sigma^\perp M}_\infty\leq tb\quad\text{and}\quad \norm{MP_\Sigma^\perp e^{t\cL}P_\Sigma M}_\infty\leq tc_{p}b.
			\end{equation*}
			By the same expansion of $(Me^{\frac{t}{n}\cL})^n=(P_\Sigma Me^{\frac{t}{n}\cL}+M^\perp e^{\frac{t}{n}\cL})^n$ as in the proof of \Cref{prop:spectral-gap-uniform-norm-power-convergence},
			\begin{equation*}
				\begin{aligned}
					\norm{\left(Me^{\frac{t}{n}\cL}\right)^n-\left(P_\Sigma Me^{\frac{t}{n}\cL}\right)^n}_\infty\leq c_2\tilde{\delta}^n+\frac{1}{n}\frac{tbc_{p}c_2(2+tbc_{p}c_2)(\tilde{\delta}-\tilde{\delta}^{n})}{1-\tilde{\delta}}e^{2tbc_{p}c_2}.
				\end{aligned}
			\end{equation*}
			Finally, \Cref{lem:properties-semigroups} shows 
			\begin{equation*}
				\norm{P_\Sigma Me^{\frac{t}{n}\cL}(\1-P_\Sigma)}_\infty=\frac{t}{n}\norm{P_\Sigma M\cL\int_{0}^{1} e^{\tau\frac{t}{n}\cL}(\1-P_\Sigma)d\tau}_\infty\leq\frac{tbc_{p}}{n}
			\end{equation*}
			so that
			\begin{equation*}
				\norm{\left(Me^{\frac{t}{n}\cL}\right)^n-\left(P_\Sigma Me^{\frac{t}{n}\cL}P_\Sigma\right)^n}\leq c_2\tilde{\delta}^n+\frac{tb}{n}+\frac{1}{n}\frac{tbc_{p}c_2(2+tbc_{p}c_2)(\tilde{\delta}-\tilde{\delta}^{n})}{1-\tilde{\delta}}e^{2tbc_{p}c_2},
			\end{equation*}
			which finishes the proof.
		\end{proof}

	\subsubsection*{Upper bound on \texorpdfstring{\Cref{eq:proofthm2-term2}}{Equation (2.1)}}\label{subsec:proofthm2-term2}~\\
		As in \Cref{lem:proofthm1-term2}, we apply the modified Chernoff Lemma (\Caref{lem:improved-chernoff}) to upper bound the second term (\ref{eq:proofthm2-term2}). 
		However, our proof strategy includes two crucial improvements compared to Theorem 3 in \cite{Becker.2021}. Firstly, we show that the spectrum of the perturbed contraction is upper semicontinuous under certain assumptions on $M$ and the $C_0$-semigroup. Therefore, we can use the holomorphic functional calculus and apply the modified Chernoff Lemma with respect to each eigenvalue separately, which allows us to achieve the optimal convergence as in \Cref{lem:proofthm1-term2}. 
		\begin{lem}\label{lem:proofthm2-term2}
			Let $(\cL,\cD(\cL))$ be the generator of a $C_0$-contraction semigroup on $\cX$ and $M\in\cB(\cX)$ a contraction satisfying the same assumption as in \Cref{thm:spectral-gap-uniform}. Then, there is an $\epsilon_2>0$, and a $\tilde{d}_1\geq0$ so that for all $t\geq0$ and $n\in\N$ satisfying $t\in[0,n\epsilon_2]$ 
			\begin{equation*}
				\norm{\left(P_\Sigma Me^{\frac{t}{n}\cL}P_\Sigma\right)^n-\sum_{j=1}^{J}\lambda_j^ne^{n\left(C_{j}(\text{\tiny{$\tfrac{t}{n}$}})-P_{j}(\text{\tiny{$\tfrac{t}{n}$}})\right)}P_{j}(\tfrac{t}{n})}_\infty\leq\frac{J}{n}e^{t\tilde{d}_1}
			\end{equation*}
		where $C_{j}(\tfrac{1}{n})\coloneqq\bar{\lambda}_jP_{j}(\tfrac{1}{n})P_\Sigma Me^{\frac{1}{n}\cL}P_\Sigma P_{j}(\tfrac{1}{n})$ and $P_\Sigma\coloneqq\sum_{j=1}^{J}P_j$.
		\end{lem}
		\begin{proof}
			By \Cref{prop:equivalence-spectral-gap}, the uniform power convergence (\ref{unifpower}) shows that the $P_j$'s are the eigenprojections of $M$ so that $P_\Sigma M=MP_\Sigma=\sum_{j=1}^{J}\lambda_jP_j$ and the spectrum $\sigma(P_\Sigma M)$ consists of\linebreak
			\begin{minipage}[c]{0.69\textwidth}
			    $J$ isolated eigenvalues on the unit circle separated by the curves $\Gamma_j:[0,2\pi]\rightarrow\C,\phi\mapsto\lambda_j+re^{i\phi}$ (q.v.~\Caref{sec:prelim} and \Caref{fig2}) with radius
				\begin{equation}\label{eq:defr}
					r\coloneqq\min_{i\neq j}\left\{\frac{\abs{\lambda_i-\lambda_j}}{3}\right\}.
				\end{equation}
				Note that we use the curve interchangeably with its image and denote the formal sum of all curves around the eigenvalues $\{\lambda_j\}_{j=1}^J$ by $\Gamma$. In the following, we define the vector-valued function:
				\begin{equation}\label{eq:defB(t)}
					\begin{aligned}
						s\mapsto P_\Sigma M+sB(s)&\coloneqq P_\Sigma M+sP_\Sigma M\cL\int_0^1e^{s\tau \cL}P_\Sigma d\tau\\
						&\;=P_\Sigma Me^{s\cL}P_\Sigma\,.
					\end{aligned}
				\end{equation}
			\end{minipage}
			\begin{minipage}[c]{0.3\textwidth}
				\begin{center}
					\begin{tikzpicture}
						\draw[dashed,tumivory] (0,0) ellipse (1.4cm and 1.4cm);
						\draw[dashed,mygreen] (0,0) ellipse (1.15cm and 1.15cm);
						\filldraw (0,0) ellipse (0.03cm and 0.03cm);
						
						\filldraw (15:1.4) ellipse (0.03cm and 0.03cm);
						\draw[dashed,myorange] (15:1.4) ellipse (0.2cm and 0.2cm);
						\draw[myorange] (15:1.9) node {$\Gamma_1$};
						
						\filldraw (70:1.4) ellipse (0.03cm and 0.03cm);
						\draw[dashed,myorange] (70:1.4) ellipse (0.2cm and 0.2cm);
						\draw[myorange] (70:1.9) node {$\Gamma_2$};
						
						\filldraw (110:1.4) ellipse (0.03cm and 0.03cm);
						\draw[dashed,myorange] (110:1.4) ellipse (0.2cm and 0.2cm);
						\draw[myorange] (110:1.9) node {$\Gamma_3$};
						
						\filldraw (130:1.4) ellipse (0.03cm and 0.03cm);
						\draw[dashed,myorange] (130:1.4) ellipse (0.2cm and 0.2cm);
						\draw[myorange] (130:1.9) node {$\Gamma_4$};
						
						\filldraw (170:1.4) ellipse (0.03cm and 0.03cm);
						\draw[dashed,myorange] (170:1.4) ellipse (0.2cm and 0.2cm);
						\draw[myorange] (170:1.9) node {$\Gamma_5$};
						
						\filldraw (195:1.4) ellipse (0.03cm and 0.03cm);
						\draw[dashed,myorange] (195:1.4) ellipse (0.2cm and 0.2cm);
						\draw[myorange] (195:1.9) node {$\Gamma_6$};
						
						\filldraw (260:1.4) ellipse (0.03cm and 0.03cm);
						\draw[dashed,myorange] (260:1.4) ellipse (0.2cm and 0.2cm);
						\draw[myorange] (260:1.9) node {$\Gamma_7$};
						
						\filldraw (300:1.4) ellipse (0.03cm and 0.03cm);
						\draw[dashed,myorange] (300:1.4) ellipse (0.2cm and 0.2cm);
						\draw[myorange] (300:1.9) node {$\Gamma_8$};
						
						\filldraw (350:1.4) ellipse (0.03cm and 0.03cm);
						\draw[dashed,myorange] (350:1.4) ellipse (0.2cm and 0.2cm);
						\draw[myorange] (350:1.9) node {$\Gamma_J$};
						
						\draw [dotted,myorange,domain=310:340] plot ({1.9*cos(\x)}, {1.9*sin(\x)});
						
						%\draw[<->] (-0.025,0.0433)--(-0.475cm,0.8227cm);
						\draw[<->] (0.0,0.05)--(0,0.95);
						\draw[<->,mygreen] (0.0433,-0.025)--(0.9526cm,-0.55cm);
						
						\fill[pattern=my north east lines] (0,0) ellipse (1cm and 1cm);
						\filldraw[white] (-0.35,0.5) ellipse (0.15cm and 0.25cm);
						\draw (-0.35,0.5) node {$\delta$};
						\filldraw[white] (0.35,-0.5) ellipse (0.15cm and 0.25cm);
						\draw[mygreen] (0.35,-0.5) node {$\tilde{\delta}$};
						\filldraw[white] (-1.06,-1.06) ellipse (0.15cm and 0.3cm);
						\draw (-1.06,-1.06) node [mygreen]{$\gamma$};
						
						\draw (1.5,1.2) node {$\sigma(M)$};
						%\draw (0,0) node [fill=white]{$\spec(B)\backslash\{1\}$};
						%\draw (1.9,1.3) node {$\mathbb{C}$};
					\end{tikzpicture}
					\captionof{figure}{}\label{fig2}
				\end{center}
			\end{minipage}\\[1ex]
			Since $M\cL$ is bounded, the defined vector-valued map converges in the uniform topology to $P_\Sigma M$. Moreover, \Cref{lem:properties-semigroups} shows that $s\mapsto B(s)$ is uniformly bounded and continuous in $s=0$ because
			\begin{equation}\label{eq:proofthm2-continuity-b}
				\norm{B(s)-B(0)}_\infty=s\norm{P_\Sigma M\cL\int_{0}^{1}\int_0^{1}\tau_1e^{\tau_1\tau_2s\cL}\cL P_\Sigma d\tau_2d\tau_1}_\infty\leq s\frac{b^2}{2}
			\end{equation}
			where we have used the assumption $\|P_\Sigma\|_\infty\leq 1$. Then, we can apply \Cref{lem:quantitative-appro-riesz-projection} which shows that there exists an $\epsilon_2>0$ such that for all $s\in[0,\epsilon_2]$
			\begin{equation*}
				P_j(s)\coloneqq\frac{1}{2\pi i}\oint_{\Gamma_j}R(z,P_\Sigma M+sB(s))dz
			\end{equation*}
			defines the perturbed spectral projection w.r.t.~$\lambda_j$. Next, let $t\geq0$, $n\in\N$ such that $t\in[0,n\epsilon_2]$. By \Cref{lem:quantitative-appro-riesz-projection} and \Cref{eq:proofthm2-continuity-b}, the perturbed spectral projection can be approximated by
			\begin{align}
			    \norm{P_j(\tfrac{t}{n})-P_j}_\infty&\leq\frac{t}{n}R_jb\left(d_j-\frac{1}{2}\right)r\leq\frac{t}{n}R_jbd_jr\eqqcolon\frac{t}{n}v_j \label{eq:profthm2term2-approx-constant0}\\ 
			    \norm{P_j(\tfrac{t}{n})-P_j-\tfrac{t}{n}P_j'}_\infty&\leq\frac{t^2}{n^2}R_j^2b^2rd_j\label{eq:profthm2term2-approx-constant}
			\end{align}
			where $R_j\coloneqq\sup_{z\in\Gamma_j}\|R(z,P_\Sigma M)\|_\infty$, $d_j\coloneqq R_j\inf_{z\in\Gamma_j}\frac{2+2|z|^2}{1+2|z|^2}+\frac{1}{2}$, and we use that $|\Gamma_j|=2\pi r$. Note that the defined $d_j$ is not exactly the $d$ in \Cref{lem:quantitative-appro-riesz-projection}. Moreover, note that $\|P_j(\tfrac{t}{n})\|_\infty\leq d_jr$ and $\|P'_j\|_\infty\leq R_j^2br$. By the spectral decomposition,
			\begin{equation}\label{eq:proof-thm2-spectral-decomp}
				\left(P_\Sigma Me^{\frac{t}{n}\cL}P_\Sigma \right)^n=\sum_{j=1}^{J}\left(P_j\left(\tfrac{t}{n}\right)P_\Sigma Me^{\frac{t}{n}\cL}P_\Sigma P_j\left(\tfrac{t}{n}\right)\right)^n.
			\end{equation}
			Next, we aim at applying the modified Chernoff lemma \ref{lem:improved-chernoff} to $C_{j}(\tfrac{t}{n})\coloneqq\bar{\lambda}_jP_{j}(\tfrac{t}{n})P_\Sigma Me^{\frac{1}{n}\cL}P_\Sigma P_{j}(\tfrac{1}{n})$ for all $j\in\{1,..,J\}$, which has to be adapted since it is no longer clear that $\|C_j(\tfrac{1}{n})\|_\infty=1$. We start by bounding the difference in norm between $C_t(\tfrac{t}{n})$ and $P_j(\tfrac{t}{n})$. By the fundamental theorem of calculus and the facts $\|P_j(\tfrac{t}{n})\|_\infty\leq d_jr$, $\|P_{\Sigma}M\cL\|_\infty\leq b$, $\|e^{s\cL}\|_\infty\leq 1$, $|\lambda_j|=1$
			\begin{equation}\label{eq:proofthm2term2-approx1}
				\begin{aligned}
					\norm{C_{j}(\tfrac{t}{n})-P_j(\tfrac{t}{n})}_\infty&\leq\frac{t}{n}\norm{\bar{\lambda}_jP_j(\tfrac{t}{n})P_\Sigma M\cL\int_{0}^{1}e^{\frac{st}{n}\cL}P_\Sigma P_j(\tfrac{t}{n})ds}_\infty\\
					&\quad+\norm{P_j(\tfrac{t}{n})P_\Sigma MP_j(\tfrac{t}{n})-P_j(\tfrac{t}{n})}_\infty\\
					&\leq\frac{t}{n}bd_j^2r^2+\norm{\bar{\lambda}_jP_j(\tfrac{t}{n})P_\Sigma MP_j(\tfrac{t}{n})-P_j(\tfrac{t}{n})}_\infty.
				\end{aligned}
			\end{equation}
			In the next step, we focus on the second term and prove a higher order approximation then needed because in \Cref{lem:proofthm2-term3} we will reuse this calculation. In the following calculation, we use the bounds from above, in particular \Cref{eq:profthm2term2-approx-constant0} and (\ref{eq:profthm2term2-approx-constant}). Moreover, we use the product rule for derivatives, which shows $P_j'=P_jP_j'+P_j'P_j$ by $\tfrac{\partial}{\partial s}P_j(s)=\tfrac{\partial}{\partial s}P_j(s)^2$ (cf.~\cite[Lem.~3]{Mobus.2019}).
			\begin{equation}\label{eq:proofthm2term2-approx2}
			    \begin{aligned}
					&\norm{\bar{\lambda}_jP_j(\tfrac{t}{n})P_\Sigma MP_j(\tfrac{t}{n})-P_j(\tfrac{t}{n})}_\infty\\
					&\qquad=\norm{\left(P_j(\tfrac{t}{n})-P_j-\tfrac{t}{n}P_j'\right)P_\Sigma MP_j(\tfrac{t}{n})}_\infty+\norm{P_jP_j(\tfrac{t}{n})+\frac{t}{n}P_j'P_\Sigma MP_j(\tfrac{t}{n})-P_j(\tfrac{t}{n})}_\infty\\
					&\qquad\leq\begin{aligned}[t]
					    \frac{t^2}{n^2}R_j^2b^2r^2d_j^2&+\norm{P_j\left(P_j(\tfrac{t}{n})-P_j-\frac{t}{n}P_j'\right)}_\infty+\frac{t}{n}\norm{P_j'P_\Sigma M\left(P_j(\tfrac{t}{n})-P_j\right)}_\infty\\
					    &+\norm{P_j+\frac{t}{n}\left(P_jP_j'+P_j'P_j\right)-P_j(\tfrac{t}{n})}_\infty
					\end{aligned}\\
					&\qquad\leq\frac{t^2}{n^2}R_j^2b^2r^2d_j^2+2\frac{t^2}{n^2}R_j^2b^2rd_j+\frac{t^2}{n^2}R_j^3b^2d_jr^2\\
					&\qquad\leq\frac{t^2}{n^2}R_j^2b^2d_jr\left(d_jr+R_jr+2\right)
				\end{aligned}
			\end{equation}
			Combining \Cref{eq:proofthm2term2-approx1},  (\ref{eq:proofthm2term2-approx2}), and $\tfrac{t}{n}\leq\epsilon_2$ shows
			\begin{equation}\label{eq:proofthm2term2-approx}
				\begin{aligned}
					\norm{C_{j}(\tfrac{t}{n})-P_j(\tfrac{t}{n})}_\infty&\leq\frac{t}{n}bd_j^2r^2+\frac{t^2}{n^2}R_j^2b^2d_jr\left(d_jr+R_jr+2\right)\\
					&\leq\frac{t}{n}bd_jr\left(d_jr+\epsilon_2R_j^2b\left(d_jr+R_jr+2\right)\right)\eqqcolon\frac{t}{n}w_j
				\end{aligned}
			\end{equation}
			In \Cref{eq:profthm2term2-approx-constant0} and (\ref{eq:proof-thm2-chernoff}), we have proven that $\|P_j(\tfrac{t}{n})-P_j\|\leq\tfrac{t}{n}v_j$, $\|C_j(\tfrac{t}{n})-P_j(\tfrac{t}{n})\|\leq\tfrac{t}{n}w_j$ and note that $\|P_j\|_\infty=1$, $P_j(\tfrac{t}{n})C_j(\tfrac{t}{n})=C_j(\tfrac{t}{n})P_j(\tfrac{t}{n})=C_j(\tfrac{t}{n})$ holds by definition. Then we can apply the approximate version of the modified Chernoff \Cref{lem:approx-improved-chernoff} to $C_j(\tfrac{t}{n})$. This shows
			\begin{equation}\label{eq:proof-thm2-chernoff}
			    \norm{C_j(\tfrac{t}{n})^n-e^{n(C_j(\text{\tiny{$\tfrac{t}{n}$}})-P_j(\text{\tiny{$\tfrac{t}{n}$}}))}P_j(\tfrac{t}{n})}_\infty\leq\frac{t^2w_j^2}{2n}e^{t(v_j+w_j)}\leq\frac{1}{n}e^{t(v_j+2w_j)}.
			\end{equation}
			Combining \Cref{eq:proof-thm2-spectral-decomp} and (\ref{eq:proof-thm2-chernoff}) and writing out the constants $v_j$ and $w_j$ gives
			\begin{align*}
				&\norm{\left(P_\Sigma Me^{\frac{t}{n}\cL}P_\Sigma\right)^n-\sum_{j=1}^{J}\lambda_j^ne^{n\left(C_{j}(\text{\tiny{$\tfrac{t}{n}$}})-P_{j}(\text{\tiny{$\tfrac{t}{n}$}})\right)}P_{j}(\tfrac{t}{n})}_\infty\\
				&\hspace{10ex}\leq J\max\limits_{j\in\{1,..,J\}}\norm{\lambda_j^n\left(C_j(\tfrac{t}{n})\right)^n-\lambda_j^ne^{n\left(C_{j}(\text{\tiny{$\tfrac{t}{n}$}})-P_{j}(\text{\tiny{$\tfrac{t}{n}$}})\right)}P_{j}(\tfrac{t}{n})}_\infty\\
				&\hspace{10ex}\leq \frac{J}{n}e^{t(v_j+2w_j)}\\
			\end{align*}
			Finally, we define $\tilde{d}_1=\max\limits_{j\in\{1,...,J\}}v_j+2w_j$, which finishes the proof.
		\end{proof}
		
		\subsubsection*{Upper bound on \Cref{eq:proofthm2-term3}:}
		\begin{lem}\label{lem:proofthm2-term3}
			Let $(\cL,\cD(\cL))$ be the generator of a $C_0$-contraction semigroup on $\cX$ and $M\in\cB(\cX)$ a contraction with the same assumption as in \Cref{thm:spectral-gap-uniform}. Then, there exists a constant $\tilde{d}_2\geq0$ such that 
			\begin{equation*}
				\norm{\sum_{j=1}^{J}\lambda_j^ne^{n\left(C_{j}(\text{\tiny{$\tfrac{t}{n}$}})-P_{j}(\text{\tiny{$\tfrac{t}{n}$}})\right)}P_{j}(\tfrac{t}{n})-\sum_{j=1}^{J}\lambda_j^ne^{tP_j\cL P_j}P_j}_\infty\leq\frac{J}{n}e^{t\tilde{d}_2}\max\limits_{s\in[0,1]}\norm{e^{stP_j\cL P_j}}_\infty.
			\end{equation*}
		\end{lem}
		\begin{proof}
			For ease of notation, we absorb the time parameter $t$ into the generator $\cL$ and $b$. In order to prove the convergence of the generator, \Cref{eq:proofthm2term2-approx2} proves:
			\begin{equation*}
				\begin{aligned}
					\norm{n\left(C_{j}(\tfrac{1}{n})-P_j(\tfrac{1}{n})\right)-P_j\cL P_j}_\infty&\leq\norm{\bar{\lambda}_jP_j(\tfrac{1}{n})P_\Sigma M\cL\int_{0}^{1}e^{\frac{s}{n}\cL}P_\Sigma P_j(\tfrac{1}{n})ds-P_j\cL P_j}_\infty\\
					&\quad+\frac{1}{n}R_j^2b^2d_jr\left(d_jr+R_jr+2\right),
				\end{aligned}
			\end{equation*}
		where $R_j\coloneqq\sup_{z\in\Gamma_j}\|R(z,P_\Sigma M)\|_\infty$, $d_j\coloneqq R_j\inf_{z\in\Gamma_j}\frac{2+2|z|^2}{1+2|z|^2}+\frac{1}{2}$, and $r$ is the radius of the curves $\Gamma_j$ defined in \Cref{eq:defr}. Then, we apply \Cref{lem:properties-semigroups} and \Cref{lem:quantitative-appro-riesz-projection} on the first term:
			\begin{align*}
				&\norm{\bar{\lambda}_jP_j(\tfrac{1}{n})P_\Sigma M\cL\int_{0}^{1}e^{\frac{\tau_1}{n}\cL}P_\Sigma P_j(\tfrac{1}{n})d\tau_1- P_j\cL P_j}_\infty\\
			    &\qquad \leq\frac{1}{n}\norm{\bar{\lambda}_jP_j(\tfrac{1}{n})P_\Sigma M\cL\int_{0}^{1}\int_{0}^{1}\tau_1e^{\frac{\tau_1\tau_2}{n}\cL}\cL P_\Sigma P_j(\tfrac{1}{n})d\tau_2d\tau_1}_\infty\\		&\qquad\quad+\norm{\bar{\lambda}_jP_j(\tfrac{1}{n})P_\Sigma M\cL P_\Sigma P_j(\tfrac{1}{n})-P_j\cL P_j}_\infty\\
			    &\qquad\leq\frac{1}{2n}b^2d_j^2r^2+\norm{\left(P_j(\tfrac{1}{n})-P_j\right)P_\Sigma M\cL P_\Sigma P_j(\tfrac{1}{n})}_\infty+\norm{P_j\cL \left(P_j(\tfrac{1}{n})-P_j\right)}\\
			    &\qquad\leq\frac{1}{n}b^2d_jr\left(\frac{1}{2}d_jr+R_jd_jr+R_j\right)\,,
			\end{align*}
			where we used \Cref{eq:profthm2term2-approx-constant0} in the last step and the assumption that $M\cL$ and $\cL P_{\Sigma}$ are bounded by $b$ and all the inequalities discussed before \Cref{eq:proofthm2term2-approx}. In combination with \Cref{lem:integral-equation-semigroups} 
			\begin{align*}
				&\norm{e^{n\left(C_{j}(\text{\tiny{$\tfrac{1}{n}$}})-P_{j}(\text{\tiny{$\tfrac{1}{n}$}})\right)}-e^{P_j\cL P_j}}_\infty\\
				&\qquad\leq\max\limits_{s\in[0,1]}\norm{e^{sP_j\cL P_j}}_\infty\norm{e^{(1-s)n\left(C_{j}(\text{\tiny{$\tfrac{1}{n}$}})-P_{j}(\text{\tiny{$\tfrac{1}{n}$}})\right)}}_\infty\norm{n\left(C_{j}(\tfrac{1}{n})-P_j(\tfrac{1}{n})\right)-P_j\cL P_j}_\infty\\
				&\qquad\leq\frac{1}{n}\max\limits_{s\in[0,1]}\norm{e^{sP_j\cL P_j}}_\infty e^{w_j}b^2d_jr\left(R_j^2\left(d_jr+R_jr+2\right)+\frac{1}{2}d_jr+R_jd_jr+R_j\right)
			\end{align*}
			for all $j\in\{1,..J\}$ and where $w_j$ is defined in \Cref{eq:proofthm2term2-approx}. With one more application of \Cref{eq:profthm2term2-approx-constant0}, this shows
			\begin{equation*}
				\norm{\sum_{j=1}^{J}\lambda_j^ne^{n\left(C_{j}(\text{\tiny{$\tfrac{1}{n}$}})-P_{j}(\text{\tiny{$\tfrac{1}{n}$}})\right)}P_{j}(\tfrac{1}{n})-\sum_{j=1}^{J}\lambda_j^ne^{P_j\cL P_j}P_j}_\infty\leq\frac{J}{n}e^{t\tilde{d}_2}\max\limits_{s\in[0,1]}\norm{e^{stP_j\cL P_j}}_\infty
			\end{equation*}
			where we choose $\tilde{d}_2\geq0$ appropriately and redefine $\cL$ by $t\cL$ and $b$ by $bt$.
		\end{proof}
	
		\subsubsection*{End of the Proof of \Cref{thm:spectral-gap-uniform}:}
		Finally, we combine the upper bounds found in the lemmas in order to finish the proof of \Cref{thm:spectral-gap-uniform}.
		\begin{proof}[Proof of \Cref{thm:spectral-gap-uniform}]
			\Cref{lem:proofthm2-term1}, \ref{lem:proofthm2-term2}, and \ref{lem:proofthm2-term3} show for all $t\in[0,n\epsilon]$ with $\epsilon\coloneqq \min\{\epsilon_1,\epsilon_2\}$
			\begin{flalign*}
				&&\norm{\left(Me^{\frac{t}{n}\cL}\right)^n-\sum_{j=1}^{J}\lambda_j^ne^{tP_j\cL P_j}P_j}_\infty&\leq c_2\tilde{\delta}^n+\frac{tb}{n}+\frac{tb}{n}\frac{c_{p}c_2(2+tbc_{p}c_2)(\tilde{\delta}-\tilde{\delta}^{n})}{1-\tilde{\delta}}e^{2tbc_{p}c_2}&&\text{(\Caref{lem:proofthm2-term1})}\\
				&& &\quad+\frac{J}{n}e^{t\tilde{d}_1}&&\text{(\Caref{lem:proofthm2-term2})}\\
				&& &\quad+\frac{J}{n}e^{t\tilde{d}_2}\max\limits_{s\in[0,1]}\norm{e^{stP_j\cL P_j}}_\infty&&\text{(\Caref{lem:proofthm2-term3})}.
			\end{flalign*}
			For an appropriate constant $c_1\geq0$, we finish the proof of \Cref{thm:spectral-gap-uniform}.
		\end{proof}

	\section{Examples}\label{sec:applications}
		In this section, we present two classes of examples, which illustrate the range of applicability of our results. In the examples, we denote by  $\rho,\sigma$ quantum states. 
		\begin{ex}[Finite dimensional quantum systems]\label{ex:finite-dim}
			We choose $\cX=\cB(\cH)$ to be the algebra of linear operators over a finite dimensional Hilbert space $\cH$ endowed with the trace norm $\|x\|_1=\tr|x|$, $M:\cB(\cH)\to \cB(\cH)$ a quantum channel, i.e.~a completely positive, trace preserving linear map, and $\cL$ the generator of a semigroup of quantum channels over $\cB(\cH)$, also known as a quantum dynamical semigroup. In finite dimension, it is know that every quantum channel is a contraction \cite[Cor.~3.40]{Watrous.2018}, the spectrum includes the eigenvalue $1$ \cite[Thm.~3]{Wolf.2010}, and every linear operator in finite dimension has a discrete spectrum. Moreover, the nilpotent part of a quantum channel is zero \cite[Lem.~A.1]{Hasenohrl.2020}. Therefore, there exist $\delta\in(0,1)$, $\tilde{c}>0$, and a set of eigenvalues and projections $\{\lambda_j,P_j\}_{j=1}^J$ such that for all $n\in\N$
			\begin{equation*}
				\norm{M^nx-\sum_{j=1}^J\lambda_jP_jx}_1\leq\tilde{c}\delta^n\norm{x}_1
			\end{equation*}
		    Note that the assumptions on the semigroup are satisfied due to the finiteness of the system and the contraction property of the $P_j$ must be assumed additionally.
		\end{ex}
		In the following example class, we calculate $\delta$ directly.
		\begin{ex}[{Power convergence via strong data processing inequalities}]
			{As in \Cref{ex:finite-dim}, we define $\cX=\cB(\cH)$ endowed with the trace norm $\|x\|_1$, $M:\cB(\cH)\to \cB(\cH)$ a quantum channel, and $\cL$ the generator of a quantum dynamical semigroup.}
			Here, we further assume the existence of a projection $P:\cB(\cH)\to \mathcal{N}$ onto a subalgebra $\mathcal{N}\subset \cB(\cH)$ with $MP=PM$ and such that the following \textit{strong data processing inequality} holds for some $\hat{\delta}\in(0,1)$: for all states $\rho\in\cX$,
			\begin{align}\label{SDPI}
				D(M(\rho)\|M\circ P(\rho))\le \hat{\delta}\,D(\rho\|P(\rho))\,,
			\end{align}
			where we recall that the relative entropy between two quantum states, i.e.~positive, trace-one operators on $\cH$, is defined as $D(\rho\|\sigma):=\tr[\rho\log\rho-\rho\log\sigma]$, whenever $\operatorname{supp}(\rho)\subseteq\operatorname{supp}(\sigma)$. Equation \eqref{SDPI} was recently shown to hold under a certain detailed balance condition for $M$ in \cite{gao2021complete}: there exists a full-rank state $\sigma$ such that for any two $x,y\in\cB(\cH)$, $$\tr[\sigma\,x^*M^*(y)]=\tr[\sigma\,M^*(x^*)y].$$ Here $x^*$, resp. $M^*$, denotes the adjoint of $x$ w.r.t.~the inner product on $\cH$, resp. the adjoint of $M$ w.r.t. the Hilbert-Schmidt inner product on $\cB(\cH)$. In finite dimensions, the quantity $\sup_\rho\,D(\rho\|P(\rho))<\infty$ is called the Pimsner-Popa index of $P$ \cite{pimsner1986entropy}. Using Pinsker's inequality, we see that the assumption of \Cref{thm:spectral-gap-uniform} is satisfied: for all $x=x^*\in\cB(\cH)$ with $\|x\|_1\le 1$ and  decomposition $x=x_+-x_-$ into positive and negative parts and corresponding states $\rho_{\pm}=x_{\pm}/\tr[x_{\pm}]$,
			\begin{align*}
				\|(M^n-M^n\circ P)(x)\|_1& \le \tr[x_+]\,\|(M^n-M^n\circ P)(\rho_+)\|_1+ \tr[x_-]\,\|(M^n-M^n\circ P)(\rho_-)\|_1\\
				&= \|x\|_1\,\sqrt{2}\,\max_{\rho\in\{\rho_+,\rho_-\}}D(M^n(\rho)\|M^n\circ P(\rho))^{\frac{1}{2}}\\
				&\le \sqrt{2}\,\sup_\rho\,D(\rho\|P(\rho))^{\frac{1}{2}}\,\hat{\delta}^{\frac{n}{2}}\eqqcolon \tilde{c}\,\delta^n\,.
			\end{align*}
			{Then, we can apply \Cref{cor:explicit-bound-prop} which proves that there is an $\epsilon>0$ such that for all $n\in\N$, $t\in[0,n\epsilon]$, and $\tilde{\delta}\in(\delta,1)$
			\begin{equation*}
				\norm{\left(Me^{\frac{t}{n}\cL}\right)^n-e^{tP\cL P}P}_\infty=\cO\left(\frac{e^{t\|\cL\|_\infty}}{n}+\frac{\tilde{\delta}}{n}e^{\frac{8\tilde{c}t\norm{\cL}_\infty}{\tilde{\delta}-\delta}}\right)
			\end{equation*}
			for $n\rightarrow\infty$.}
		\end{ex}
		\begin{ex}[Infinite dimensional quantum systems and unbounded generators]
			Here, we pick $\cH=L^2(\mathbb{R})$, denote by $I$ the identity operator on $\cH$, let $\sigma$ be a quantum state on $\cH$ and $M$ a generalized depolarizing channel of depolarizing parameter $p\in (\frac{1}{2},1)$ and fixed point $\sigma$:
			\begin{align}
				M(\rho):=(1-p)\rho +p\tr(\rho)\,\sigma\,.
			\end{align}
			It is clear by convexity that $M$ satisfies the uniform strong power convergence with projection $P(\rho)=\tr(\rho)\,\sigma$ and parameter $\delta=2(1-p)<1$,
			\begin{equation*}
			    \norm{M(\rho)-P(\rho)}_1=(1-p)\norm{\rho -\tr(\rho)\,\sigma}_1\leq2(1-p)\norm{\rho}_1.
			\end{equation*}
			Let $\cL$ be a generator of a $C_0$-contraction semigroup such that $\sigma\in\cD(\cL)$. For example, let $\cH$ be the Fock-space spanned by the Fock basis $\{\ket{0},\ket{1},\ket{2},...\}$, $a$ and $a^\dagger$ be the annihilation and creation operator defined by $a\ket{0}=0$, $a\ket{j}=\sqrt{j}\ket{j-1}$ for all $j\in\N_{\geq 1}$, and $a^\dagger\ket{j}=\sqrt{j+1}\ket{j+1}$ for all $j\in\N_{\geq0}$. Then, define $e^{t\cL}(\rho)\coloneqq e^{-itH}\rho e^{itH}$ where $H=a^\dagger a+\tfrac{1}{2}$ ($H=-\Delta+x^2$) is the Hamiltonian of the harmonic oscillator as in \cite{Becker.2021} and 
			\begin{equation*}
			    \sigma=\frac{1}{3}(\ket{0}\bra{0}+\ket{1}\bra{1}+\ket{2}\bra{2})+\frac{1}{10}(\ket{0}\bra{1}+\ket{1}\bra{0}).
			\end{equation*}
			Then, we have that for all $t\ge 0$: 
			\begin{align}\label{upperbound1}
				\|(\1-P)e^{t\cL}P\|_{1\to 1}=\sup_{\|x\|_1\le 1}|\tr(x)|\,\|e^{t\cL}(\sigma)-\sigma\|_1\le t\,\|\cL(\sigma)\|_1\,.
			\end{align}
			Moreover, by duality and the unitality of the maps $e^{t\cL^*}$ we have that
			\begin{align}\label{upperbound2}
				\|Pe^{t\cL}(\1-P)\|_{1\to 1}=\|(\1-P)^*e^{t\cL^*}P^*\|_{\infty\to\infty}=\sup_{\|y\|\le 1}|\tr(\sigma y)|\|(\1-P)(I)\|=0\,,
			\end{align}
			Therefore, the assumptions of \Cref{thm:spectral-gap} are satisfied and we find the convergence rate $\mathcal{O}(n^{-1})$. Interestingly, this answers a conjecture of \cite[Ex.~3, 5]{Becker.2021} for the Hamiltonian evolution generated by the one-dimensional harmonic oscillator. There, the authors had numerically guessed the optimal rate which we prove here. However their analytic bounds could only provide a decay of order $\mathcal{O}(n^{-\frac{1}{4}})$ (q.v.~remark after \Caref{lem:chernoff}) and for a restriction of $H$ to a finite dimensional stable subspace, which effectively assumed the boundedness of the generator.
		\end{ex}
        The depolarizing noise considered in the previous example is artificial. In an infinite dimensional bosonic system, a more natural model of noise is the photon loss channel, which we consider in the next example.

		\begin{ex}[Bosonic beam-splitter]
            We define the bosonic one-mode system by the algebra generated by the creation and annihilation operators $a^*$ and $a$ which satisfy the canonical commutation relation (CCR):
			\begin{equation*}
				[a,a^*]=\1\,.
			\end{equation*}
			The associated Fock basis $\{\ket{0},\ket{1},\ket{2},...\}$ is orthonormal and defined by 
			\begin{equation*}
			a^*\ket{j}=\sqrt{j+1}\ket{j+1}\quad\text{and}\quad a\ket{j}=\sqrt{j}\ket{j-1}
			\end{equation*}
			where the vacuum state $\ket{0}$ satisfies $a\ket{0}=0$. The Fock basis spans a Hilbert space called Fock space on which the operators from the CCR algebra are defined. A bosonic quantum state is a semidefinite operator in the CCR algebra with trace 1. 
			A bosonic $2$-mode system is defined by the CCR-algebra generated by $\{a,b,a^*,b^*\}$, which satisfy, additionally to the canonical commutation relation, $[a,b]=0$. Next, we consider the \textit{quantum beam-splitter} for $\lambda\in [0,1)$ 
			\begin{equation*}
				M_\lambda(y)\coloneqq\tr_2[U_\lambda y\otimes\sigma U_\lambda^*] \,,
			\end{equation*}
			where $\tr_2$ denotes the partial trace over the second register, $U_\lambda\coloneqq e^{(a^* b-b^* a)\operatorname{arcos}(\sqrt{\lambda})}$, an environment state $\sigma$, and $y$ an element in the CCR algebra generated by $\{a^*, a\}$. Moreover, $P(y)\coloneqq\tr[y]\sigma$ defines a projection which satisfies $PM_\lambda=M_\lambda P=P$ with the adjoint $P^*(x)=\tr[\sigma x]\1$.
            
            In order to establish e.g.~the uniform power convergence of \Cref{thm:spectral-gap} in the topology of the trace distance, we would need to consider a convergence in the form of $\|M_\lambda^n(\rho)-\sigma\|_1\to 0$ in the limit of large $n$ and uniformly in the initial state $\rho$. Such property is notoriously hard to prove even in the classical setting \cite{PW16}. Instead, we will consider a different metric on the set of quantum states which turns out to be more easy to work with.
			
			We write $\cB_{{N}}$ for the linear space of all ${N}$-bounded operators, where $ N= a^\dagger a$ corresponds to the photon number operator. That is the vector space of linear operators $X$ on $L^2(\mathbb{R})$ such that for any $|\psi\rangle\in \operatorname{dom}(N)$, $|\psi\rangle\in\operatorname{dom}(X)$ and there are some positive constants $a,b$ such that
			\begin{align*}
			    \|X|\psi\rangle\|\le a\|N|\psi\rangle\|+b\|\psi\|\,.
			\end{align*}
			 We define the \textit{Bosonic Lipschitz constant} of a $X\in \cB_{{N}}$ as \cite{Cambyse.2021}
            \begin{align}
                \|\nabla X\|^2 :=  \sup_{|\psi\rangle,|\varphi\rangle}\,|\langle \psi|[a,X]|\varphi\rangle |^2+|\langle \psi|[a^*,X]|\varphi\rangle|^2\,,\nonumber
            \end{align}
            where the suppremum is over all pure states $|\psi\rangle,|\varphi\rangle\in\operatorname{dom}({N})$ of norm $1$. By duality, we then define the \textit{Bosonic Wasserstein norm} of a linear functional $f$ over $\mathcal{B}_N$ with $f(\1)=0$ as
            \begin{align*}
                \|f\|_{W_1}:=\sup_{\|\nabla X\|\le 1}\,\big|f(X)\big|\,.
            \end{align*}
            where the supremum is over all ${N}$-bounded, self-adjoint operators $X$. We then choose our Banach space $\cX$ as the closure of the set of such linear functionals such that $\|f\|_{W_1}<\infty$. In particular, whenever $f\equiv f_{\rho-\sigma}$ is defined in terms of the difference between two quantum states $\rho,\sigma$ as $f_{\rho-\sigma}(X)=\tr((\rho-\sigma) X)$, we denote the Wasserstein distance associated to the norm $\|.\|_{W_1}$ as (see also \cite{Cambyse.2021}):
            \begin{align*}
                W_1(\rho,\sigma):=\|f_{\rho-\sigma}\|_{W_1}\,.
            \end{align*}			
			These definitions extend the classical Lipschitz constant $\|\nabla f\|:=\sup_{x\in\mathbb{R}^2}|\nabla f(x)|$ of a real, continuously differentiable function $f$ of $2$ variables as well as the dual Wasserstein distance over probability measures on $\mathbb{R}^2$. 
			
			In order to relate the Wasserstein distance to the statistically more meaningful trace distance, we seek for an upper bound on the trace distance in terms of $W_{1}$. By duality of both metrics, this amounts to finding an upper bound on the Lipschitz constant $\|\nabla X\|$ of any bounded operator $X$ in terms of its operator norm $\|X\|_\infty$. However, a bound of that sort does not exist (as classically, one can easily think of bounded observables which are not \textit{smooth}). In the classical setting, the problem can be handled by first \textit{smoothing} the function $f$, e.g. by convolving it with a Gaussian density $g$. In that case, one proves that there exists a finite constant $C>0$ such that $\|\nabla (f\ast g)\|\le C\|f\|_\infty$. In analogy with the classical setting, we can prove that for any two states $\rho_1,\rho_2$ and $\lambda\in [0,1)$ (see also \cite[Proposition 6.4]{Cambyse.2021}),
			\begin{align}\label{tracetoWass}
			    \|M_\lambda(\rho_1-\rho_2)\|_1\le \,C\, W_{1}(\rho_1,\rho_2)\,,
			\end{align}
			where $C^2:=(\|[a,\sigma]\|_1^2+\|[a^*,\sigma]\|_1^2)\lambda {(1-\lambda)^{-1}}$.
			
            With a slight abuse of notations, we also write $M_\lambda(f)$ for $f\circ \mathcal{B}^*_\lambda$. It remains to prove the uniform power convergence. Proposition 6.2 from \cite{Cambyse.2021} gives 
			\begin{align*}
				\norm{M_\lambda(f)}_{W_1}&=\sup_{\|\nabla X\|\leq1}\abs{f\circ M_\lambda^*(X)}\\
				&=\sup_{\|\nabla X\|\leq1}\abs{f\left(\frac{M_\lambda^*(X)}{\|\nabla M_\lambda^*(X)\|}\right)}\|\nabla M_\lambda^*(X)\|\\
				&=\sup_{\|\nabla X\|\leq1}\abs{f(X)}\sqrt{\lambda}\\
				&=\sqrt{\lambda}\,\|f\|_{W_1}
				\end{align*}
				The uniform power convergence follows by $P(f)(X)\equiv f\circ P^*(X)=\tr(\sigma X)f(\1)=0$. Moreover, the asymptotic Zeno condition (\ref{eq:thm1-asympzeno}) is satisfied if $\sigma\in\cD(\cL)$ so that \Cref{thm:spectral-gap} is applicable.
		\end{ex}
		As illustrated here, our asymptotic Zeno condition is easily verifiable and provides a rich class of examples. More examples for which our optimal convergence rate holds can be found in \cite{Becker.2021}.

	\section{Discussion and Open Questions}\label{sec:discussion}
		In this paper, we proved the optimal convergence rate of the quantum Zeno effect in two results: \Cref{thm:spectral-gap} focuses on weakening the assumptions of the $C_0$-semigroup to the so-called asymptotic Zeno condition. Hence, \Cref{thm:spectral-gap} allows strongly continuous Zeno dynamics which is novel for open systems. In \Cref{thm:spectral-gap-uniform} instead, we weaken the assumption on $M$ to the uniform power convergence as in \cite[Thm.~3]{Becker.2021}. Additionally, we presented an example which shows the optimality of the achieved convergence rate. This brings up the question whether our assumptions are optimal and how the assumption on the contraction correlates with the assumption on the $C_0$-semigroup. For example, is it possible to weaken the uniform power convergence in \Cref{thm:spectral-gap} or \Cref{prop:spectral-gap-uniform-norm-power-convergence} to finitely many eigenvalues on the unit circle without assuming stronger assumption on the semigroup? Following our proof strategy (q.v.~\Caref{lem:proofthm1-term1}), this question is related to the conjecture of a generalized version of \textit{Trotter's product formula for finitely many projections} under certain assumptions on the generator,
		\begin{equation*}
			\norm{\left(\sum_{j=1}^{J}\lambda_jP_je^{\frac{1}{n}\cL}\right)^nx-\sum_{j=1}^{J}\left(\lambda_jP_je^{\frac{1}{n}\cL}\right)^nx}\overset{?}{=}\cO\left(\frac{1}{n}(\norm{x}+\norm{\cL x}+\|\cL^2 x\|)\right),
		\end{equation*}		
		Another line of generalization would be to weaken the assumption on $M$ to the strong topology as in Theorem 2 in \cite{Becker.2021}. There, the authors assume that $M^n$ converges to $P$ in the strong topology and that the semigroup is uniformly continuous. It would be interesting to know whether an extension to $C_0$-semigroups is possible. 
		Finally, another important line of generalization would be to extend our results to time-dependent semigroups as in \cite{Mobus.2019}.

	\emph{Acknowledgments:} We would like to thank Michael Wolf and Markus Hasen\"ohrl for their support on this project. Moreover, we would like to thank Valentin A. Zagrebnov for his helping out with questions on the Chernoff Lemma and the anonymous reviewers for their constructive and detailed feedback. T.M. and C.R. acknowledge the support of the Munich Center for Quantum Sciences and Technology, and C.R. that of the Humboldt Foundation.
	
	%References
	\setlength{\bibitemsep}{0.5ex}
	\printbibliography[heading=bibnumbered]
	\vspace{2ex}
	\addresseshere
	\appendix
	\section{Holomorphic Functional Calculus and Semicontinuous Spectra}\label{sec:appendix-holomorphic-fc-semicontinuity}
		In this section, we introduce the holomorphic functional calculus and consider some continuity properties of the spectrum of perturbed operators. These two methods are used in \Cref{prop:spectral-gap-uniform-norm-power-convergence} and \Cref{thm:spectral-gap-uniform} in the main text. Let's start with the holomorphic functional: To be precise, the holomorphic functional calculus defines $f(A)$ for $A\in\cB(\cX)$ and a function $f:\C\rightarrow\C$ holomorphic on a neighbourhood of $\sigma(A)$ in such a way that, if the spectrum is separated into two sets, then the functions restricted to the separated subsets can be treated independently. This in turn is used to define spectral projections of $A$. More details on this topic can be found in Section 2.3 of \cite{Simon.2015}. We use the holomorphic functional calculus to introduce an integral representation of $f(A)$ for a holomorphic function $f$, as well as its associated spectral projections:
		\begin{prop}[Holomorphic Functional Calculus \texorpdfstring{\cite[Thm.~2.3.1-3]{Simon.2015}}{[36, Thm.~2.3.1-3]}]\label{prop:holomorphic-functional-calculus}
			Let $A\in\cB(\cX)$, $\Gamma:[0,2\pi]\rightarrow\C$ be a curve around $\sigma(A)$, and $f$ be a function that is holomorphic in the neighbourhood of $\sigma(A)$. Then,
			\begin{equation*}
				f(A)\coloneqq\frac{1}{2\pi i}\oint_{\Gamma}f(z)R(z,A)dz
			\end{equation*}
			is independent of $\Gamma$. If a subset $\tilde{\sigma}$ of the spectrum is separated by a curve $\Gamma$, the spectral projection w.r.t.~the enclosed subset is defined as
			\begin{equation*}
				P\coloneqq \frac{1}{2\pi i}\oint_{\Gamma}R(z,A)dz
			\end{equation*}
			and satisfies $P^2=P$ and $PA=AP$. If $\tilde{\sigma}=\{\lambda_0\}$ is an isolated eigenvalue, the following nilpotent operator
			\begin{equation*}
				N\coloneqq \frac{1}{2\pi i}\oint_{\Gamma}(z-\lambda_0)R(z,A)dz,
			\end{equation*}
			satisfies $NP=PN=N$ and $AP=\lambda_0P+N$.
		\end{prop}
		In our case, $M$ has finitely many isolated eigenvalues on the unit circle. However, if the semigroup perturbs those eigenvalues, the following \namecref{lem:semicontinuity-spectrum} states a condition under which the separation of the spectrum is maintained.
		\begin{lem}[Semicontinuous Spectra \texorpdfstring{\cite[Sec.~IV.3.1]{Kato.1995}}{[25,Sec.~IV.3.1]}]\label{lem:semicontinuity-spectrum}
			Let $b\geq0$ and $M:\R_{\geq0}\rightarrow\cB(\cX)$ be a vector-valued function which satisfies
			\begin{equation*}
				\norm{M(t)-M(0)}_\infty\leq tb.
			\end{equation*}
			Moreover, assume that $\sigma(M(0))=\sigma_1\cup\sigma_2$ is a disjoint union separated by a curve $\gamma:[0,2\pi]\rightarrow\rho(M(0))$. For $\epsilon>0$ satisfying 
			\begin{equation}\label{eq:semicontinuity-time-bound}
				\epsilon<\frac{1}{2b}\inf_{z\in\gamma}(1+\abs{z}^2)^{-1}\left(1+\sup_{z\in\gamma}\norm{R(z,M(0))}_\infty^2\right)^{-\frac{1}{2}},
			\end{equation}
			the spectrum of $\sigma(M(t))$ is separated by $\gamma$ for all $t\in[0,\epsilon]$.
		\end{lem}
		\begin{proof}
			Firstly, we show that \Cref{eq:semicontinuity-time-bound} is well-defined --- i.e. the value $\epsilon>0$ exists. The compactness of the image of $\gamma$ and the continuity of the resolvent \cite[Thm.~II.1.5]{Kato.1995} show that $\sup_{z\in\gamma}\|R(z,M(0))\|_\infty<\infty$ and $\epsilon>0$ exists. Then, we can apply Theorem 3.16 from \cite[Sec.~IV]{Kato.1995}, which proves the \namecref{lem:semicontinuity-spectrum}.
		\end{proof}
		\begin{rmk*}
			It is possible to generalize the \namecref{lem:semicontinuity-spectrum} above to uniformly continuous vector-valued maps. This this case however, we loose the explicit bound on $\epsilon$.
		\end{rmk*}

	\section{Spectral Gap Assumption}\label{sec:appendix-spectral-gap}
		In the proofs of \Cref{prop:spectral-gap-uniform-norm-power-convergence} and \Cref{thm:spectral-gap-uniform}, we use the equivalence \cite[Prop.~3.1]{Becker.2021} of the uniform power convergence and the spectral gap assumption with corresponding quasinilpotent operators being zero, that is there are eigenvalues $\{\lambda_j\}_{j=1}^J\subset\partial\D_1$ and a gap $\delta\in(0,1)$ so that
		\begin{equation*}
			\sigma(M)\subset \D_\delta\cup\{\lambda_j\}_{j=1}^J.
		\end{equation*}
		Moreover, the quasinilpotent operators are assumed to be zero, that is
		\begin{equation*}
			N_j=\frac{1}{2\pi i}\oint_{\Gamma_{j}}(z-\lambda_j)R(z,M)dz=0
		\end{equation*}
		for $\Gamma_{j}:[0,2\pi)\ni\varphi\mapsto\lambda_j+\frac{1}{2}\min\limits_{k,l\in\{1,..,J\}}\{\abs{\lambda_k-\lambda_l},1-\delta\}e^{i\varphi}$ (q.v.~\Caref{eq:quasinilpotent-operator}). Additionally, the projections defined in the uniform power convergence are equal to the spectral projections 
		\begin{equation*}
			P_j=\frac{1}{2\pi i}\oint_{\Gamma_{\lambda_j}} R(z,M)dz
		\end{equation*}
		for all $j\in\{1,...,J\}$ (q.v.~\Caref{eq:spectral-projection}). Since $N_j=0$, $P_jM=MP_j=\lambda_jP_j$ \cite[Thm.~2.3.5]{Simon.2015}. The equivalence between the uniform power convergence and the spectral gap assumption with corresponding quasinilpotent operators being zero is given in \cite[Prop.~3.1]{Becker.2021}.
		\begin{prop}[\texorpdfstring{\cite[Prop.~3.1]{Becker.2021}}{[2, Prop.~3.1]}]\label{prop:equivalence-spectral-gap}
			Let $M\in\cB(\cX)$ be a contraction, $J\in\N$, $\delta\in(0,1)$, and $\{\lambda\}_{j=1}^J\subset\partial\D_1$. Then, the following statements are equivalent:
			\begin{enumerate}
				\item The contraction $M$ satisfies the spectral gap assumption w.r.t.~$\{\lambda_j\}_{j=1}^J$, the gap $\delta$, and the corresponding quasinilpotent operators being zero. 
				\item There are projections $\{P_j\}_{j=1}^J$, $\delta\in(0,1)$ and $\tilde{c}\ge 0$ such that $MP_j=\lambda_j P_j$ for all $j=\{1,...,J\}$ and
				\begin{equation*}
					\norm{M^n-\sum_{j=1}^{J}\lambda_j^nP_j}_\infty\leq\,\tilde{c}\,\delta^n.
				\end{equation*}
				\item There are contractions $\{C_j\}_{j=1}^J$ so that $\lim\limits_{n\rightarrow\infty}\|M^n-\sum_{j=1}^{J}\lambda_j^nC_j\|_\infty=0$.
			\end{enumerate}
			If any of the conditions above holds true, $C_j=P_j$ and $P_j$ are the eigenprojectors w.r.t.~to the eigenvalue $\lambda_j$.
		\end{prop}

	\section{Approximate modified Chernoff $\sqrt{n}$ Lemma}\label{sec:appendix-chernoff}		
		\begin{lem}[Approximate Modified Chernoff Lemma]\label{lem:approx-improved-chernoff}
			Let $\epsilon>0$, $P\in\cB(\cX)$ a projection with $\|P\|_\infty=1$, and $[0,\epsilon]\ni t\mapsto P(t)\in\cB(\cX)$, $[0,\epsilon]\ni t\mapsto C(t)\in\cB(\cX)$ be two vector valued maps with the properties $P(t)^2=P(t)$, $P(t)C(t)=C(t)P(t)=C(t)$,  $\|P(t)-P\|_\infty\leq tv$, and $\|C(t)-P(t)\|_\infty\leq tw$ for some $v,w\geq0$ and all $t\in[0,\epsilon]$. Then, for all $n\in\N$ with $\frac{1}{n}\in[0,\epsilon]$
			\begin{equation*}
				\norm{C(\tfrac{1}{n})^n-e^{n(C(\text{\tiny{$\tfrac{1}{n}$}})-P(\text{\tiny{$\tfrac{1}{n}$}}))}P(\tfrac{1}{n})}\leq\frac{n}{2}e^{v+w}\norm{\left(C(\tfrac{1}{n})-P(\tfrac{1}{n})\right)^2}\leq\frac{w^2}{2n}e^{v+w}.
			\end{equation*}
		\end{lem}
		\begin{proof}
			Similar to the proof of \Cref{lem:improved-chernoff}, we define $C_t(\tfrac{1}{n})\coloneqq P(\tfrac{1}{n})+t(C(\tfrac{1}{n})-P(\tfrac{1}{n}))$ for $t\in[0,1]$ and $n\in\N$ and we use the fundamental theorem of calculus so that
			\begin{align}
				\norm{C(\tfrac{1}{n})^n-e^{n(C(\text{\tiny{$\tfrac{1}{n}$}})-P(\text{\tiny{$\tfrac{1}{n}$}}))}P(\tfrac{1}{n})}_\infty&\leq\int_{0}^{1}\norm{\frac{\partial}{\partial t}(C_t(\tfrac{1}{n})^ne^{(1-t)n(C(\text{\tiny{$\tfrac{1}{n}$}})-P(\text{\tiny{$\tfrac{1}{n}$}}))})}_\infty dt\notag\\
				&\leq n\int_{0}^{1}t\norm{C_t(\tfrac{1}{n})^{n-1}e^{(1-t)n(C(\text{\tiny{$\tfrac{1}{n}$}})-P(\text{\tiny{$\tfrac{1}{n}$}}))}}_\infty\norm{(P(\tfrac{1}{n})-C(\tfrac{1}{n}))^2}_\infty dt\notag\\
				&\leq n\norm{\left(C(\tfrac{1}{n})-P(\tfrac{1}{n})\right)^2}_\infty\int_{0}^{1}t\norm{C_t(\tfrac{1}{n})}_\infty^{n-1}e^{(1-t)n\norm{C(\text{\tiny{$\tfrac{1}{n}$}})-P(\text{\tiny{$\tfrac{1}{n}$}})}_\infty} dt\notag\\
				&\leq n\norm{\left(C(\tfrac{1}{n})-P(\tfrac{1}{n})\right)^2}_\infty\int_{0}^{1}te^{v+tw}e^{(1-t)w} dt\label{eq:proof-approx-chernoff}\\
				&=\frac{n}{2}e^{v+w}\norm{\left(C(\tfrac{1}{n})-P(\tfrac{1}{n})\right)^2}_\infty\notag\\
				&\leq\frac{w^2}{2n}e^{v+w}.\notag
			\end{align}
			In the fourth inequality (\ref{eq:proof-approx-chernoff}), we used 
			\begin{align*}
			    \norm{C_t(\tfrac{1}{n})}_\infty^{n-1}&=\norm{P+P(\tfrac{1}{n})-P+t(C(\tfrac{1}{n})-P(\tfrac{1}{n}))}_\infty^{n-1}\\
			    &\leq\left(1+\norm{P(\tfrac{1}{n})-P}_\infty+ t\norm{C(\tfrac{1}{n})-P(\tfrac{1}{n})}_\infty\right)^{n-1}\\
			    &\leq\left(1+\frac{v+tw}{n}\right)^{n}\\
			    &\leq e^{v+tw}
			\end{align*}
			which proves the \namecref{lem:improved-chernoff}.
		\end{proof}
\end{document}